\documentclass{tlp}

\usepackage{amssymb}
\usepackage{graphicx}
\usepackage{multirow}
\usepackage{xspace}
\usepackage{tikz}
\usepackage{comment}
\usepackage{stmaryrd}
\usepackage[normalem]{ulem}
\usepackage{hyperref}

\begin{document}

\lefttitle{A. Gheerbrant, L. Libkin, A. Rogova, C. Sirangelo}

\jnlPage{1}{8}
\jnlDoiYr{2021}
\doival{10.1017/xxxxx}

\title[Querying Incomplete Data]{Querying Incomplete Data : Complexity and Tractability via Datalog and First-Order Rewritings}

\begin{authgrp}
\author{Amélie Gheerbrant}
\affiliation{Université Paris cite, CNRS, IRIF, F-75013, Paris, France}
\author{Leonid Libkin}
\affiliation{School of Informatics, University of Edinburgh, 10 Crichton Street, Edinburgh EH8 9AB, UK}
\affiliation{RelationalAI, Paris, France}
\affiliation{ENS, PSL University, Paris, France}
\author{Alexandra Rogova}
\affiliation{Université Paris cite, CNRS, IRIF, F-75013, Paris, France}
\affiliation{Data Intelligence Institute of Paris (diiP) , Inria}
\author{Cristina Sirangelo}
\affiliation{Université Paris cite, CNRS, IRIF, F-75013, Paris, France}
\affiliation{DI ENS, ENS, PSL University, CNRS, Inria Paris, France}
\end{authgrp}

\history{\sub{29/03/2023;} \rev{27/09/2023;} \acc{xx xx xxxx}}

\maketitle

\begin{abstract}

To answer database queries over incomplete data the gold standard is finding certain answers:
those that are true regardless of how incomplete data is interpreted. Such answers can be found
efficiently for conjunctive queries and their unions, even in the presence of constraints. With negation added, the problem becomes intractable however.
We concentrate on the complexity of certain answers under constraints, and on
effficiently answering queries outside the usual classes of (unions) of conjunctive queries
by means of rewriting as Datalog and first-order queries. We first notice that there
are three different ways in which query answering can be cast as a decision problem. We complete the existing picture and provide precise complexity bounds on all versions of the decision
problem, for certain and best answers.
We then study a well-behaved class of queries that extends unions of conjunctive queries with a mild form of negation. We show that for them, certain answers can be expressed in Datalog with
negation, even in the presence of functional dependencies, thus making them tractable in data
complexity. We show that in general Datalog cannot be replaced by first-order logic, but without
constraints such a rewriting can be done in first-order.
The paper is under consideration in \emph{Theory and Practice of Logic Programming} (TPLP).

\end{abstract}


\newtheorem{definition}{Definition}[section]
\newtheorem{lemma}[definition]{Lemma}
\newtheorem{proposition}[definition]{Proposition}
\newtheorem{theorem}[definition]{Theorem}
\newtheorem{corollary}[definition]{Corollary}
\newtheorem{claim}[definition]{Claim}
\newtheorem{example}[definition]{Example}

\newcommand{\todo}[1]{{\color{red} TODO: {#1}}}
\newcommand{\adom}{\text{adom}}
\newcommand{\certainanswer}{{\sc CertainAnswer$_\Sigma$}}
\newcommand{\bestanswer}{{\sc BestAnswer$_\Sigma$}}
\newcommand{\xanswer}{\triangle\textsc{Answer}}
\def\conp{{\sc coNP}\xspace}
\newcommand{\const}{{\sf Const}}
\newcommand{\nulls}{{\sf Null}}
\newcommand\sem[1]{{[\![ #1 ]\!]}}
\newcommand{\vl}{\text{\sf V}}
\newcommand{\supp}{\text{\sf Supp}}
\newcommand{\qd}{{Q,D}}
\newcommand{\sless}{\lhd}
\newcommand{\LRA}{\Leftrightarrow}
\newcommand{\qda}{{Q,D,\bar a}}
\newcommand{\qdb}{{Q,D,\bar b}}
\newcommand{\best}{\text{\sf best}}
\newcommand{\x}{\text{\sf $\triangle$}}
\newcommand{\FO}{\ensuremath{\mathsf{FO}}}

\newcommand{\thetaptwo}{\text{\sc P}^{\text{\sc NP}[\log n]}}
\def\ac0{u{\sc AC$^0$}}
\def\dlog{{\sc DLogSpace}\xspace}
\def\logtwospace{{\sc DSpace($\log^2$)}\xspace}
\def\nlogspace{{\sc NLogSpace}\xspace}
\def\ptime{{\sc P}\xspace}
\renewcommand{\ptime}{\text{\sc P}}
\newcommand{\thetap}[1]{\Theta^p_{#1}}
\def\nptime{{\sc NPTime}\xspace}
\def\apspace{{\sc APSpace}}
\def\ph{{\sc PH}\xspace}
\def\pspace{{\sc PSPACE}\xspace}
\def\npspace{{\sc NPSpace}}
\def\np{{\sc NP}\xspace}
\def\classdp{{\sc DP}\xspace}
\newcommand{\sharpp}{\text{\#}\ptime}
\newcommand{\fp}{\text{\sc FP}}
\def\conp{{\sc coNP}\xspace}
\def\exptime{{\sc ExpTime}\xspace}

\newcommand{\cert}{\mathsf{cert}}
\newcommand{\chase}{\mathsf{chase}}

\newcommand{\naive}{na\"\i ve}
\newcommand{\Naive}{Na\"\i ve}

\newcommand{\certcrc}{\cert^\circ}
\newcommand{\certgr}{\cert_\gr}
\newcommand{\certcgr}{\cert^\circ_\gr}
\newcommand{\certgrc}{\certcgr}
\newcommand{\thetaanswer}{{\sc $\theta$-Comparison}}
\newcommand{\beanswer}{{\sc $\less$-Comparison}}
\newcommand{\banswer}{{\sc $\sless$-Comparison}}
\newcommand{\setbestanswer}{{\sc BestAnswer$^=$}}
\newcommand{\familybestanswer}{{\sc BestAnswer$^\in$}}
\newcommand{\setcertainanswer}{{\sc CertainAnswer$^=$}}
\newcommand{\familycertainanswer}{{\sc CertainAnswer$^\in$}}
\newcommand{\setxanswer}{\triangle\textsc{Answer}^=}
\newcommand{\familyxanswer}{\triangle\textsc{Answer}^\in}
\newcommand{\ucqneq}{\mbox{$\text{\rm UCQ}^{\neq}$}}
\newcommand{\Sep}{\text{\sf Sep}}

\newcommand{\DD}{\mathcal{D}}
\newcommand{\CC}{\mathcal{C}}
\newcommand{\XX}{\mathcal{X}}
\newcommand{\KK}{\mathcal{K}}
\newcommand{\HH}{\mathcal{H}}
\newcommand{\cS}{\mathcal{S}}
\newcommand{\M}{{\bf M}}
\newcommand{\D}{\EuFrak{D}}
\newcommand{\I}{\EuFrak{I}}
\newcommand{\Map}{\mathbb{M}}
\newcommand{\Hom}{\text{Hom}}
\newcommand{\FFF}{\mathbb{F}}

\renewcommand{\comment}[2]{#2}

\newcommand{\mycite}[1]{[\cite{#1}]}
\newcommand{\equivfo}{equivFO}

\section{Introduction}
%
%
Answering queries over incomplete databases is crucial in many
different scenarios such as data
integration \mycite{lenzerini:data-integration}, data
exchange \mycite{arenas-et-al:debook}, inconsistency
management \mycite{bertossi:book}, data
cleaning \mycite{geerts-et-all:llunatic}, ontology-based data access
(OBDA) \mycite{bienvenu-ortiz:obda}, and many others. The common
thread running through all these applications lies in
computing \emph{certain
answers} \mycite{imielinski-et-al:incomp}. Intuitively this produces
answers that are true in all {\em possible worlds}, i.e., complete
databases that 
an incomplete database represents. An incomplete database in itself is
a set of tuples with missing information, plus integrity constraints.
One can think, for example, of relations with nulls on which keys can
be specified. Then a possible world is obtained by substituting values
for nulls so that all the keys are satisfied.

The notion of certain answers is sometimes too
restrictive (for example, for some queries no answers are certain). In
that case an alternative is {\em best answers}: for them, there is
no other tuple that is an answer in more possible worlds. However,
computationally one 
encounters serious problems with both approaches. To start with,
computing certain answers and best answers is intractable for first-order
queries \mycite{abiteboul-et-al:ca,libkin:zero} (already for data complexity). Finding such answers in
restricted 
subclasses of first-order queries often relies on sophisticated
algorithms -- not naturally expressible by other queries -- that are
therefore difficult to implement in a DBMS. We know that
restricting to unions of conjunctive queries allows one to overcome this
difficulty by using \emph{na\"ive
evaluation} which computes certain answers in
polynomial time \mycite{imielinski-et-al:incomp}. This amounts to 
evaluating queries over incomplete databases as if nulls were usual
data values, thus merely using the standard database query engine to
compute certain answers.

We address these problems in the present paper whose  goal is two-fold.
\begin{enumerate}
\item We start by filling gaps in our
knowledge of the complexity of answering queries over incomplete
databases. Intractable bounds on certain and best answers cited above
were obtained under {\em different} formulations of query answering as
a decision problem. We show that there are three natural ways to
represent query answering as a decision problem, and classify the
complexity of certain and best answers for all of them.
\item We then look at a way of finding query answers by leveraging the
existing database technology, namely by finding {\em query rewritings}
which, when evaluated on the incomplete database,
give us certain answers. We show that for a class extending unions of
conjunctive queries with a form of negation (but still falling short
of all first-order queries) such rewritings can be found in Datalog
with negation, thus giving us a tractable complexity bound.
\end{enumerate}

To elaborate on the first point, the two existing decision versions of
the query answering problem are: (a) is a tuple in the answer? and (b)
is the answer a member of a given family of sets? We add a third: (c)
is the answer equal to a given set. We then prove that for certain answers
the complexity 
is \conp, $\thetaptwo$, and \classdp-complete for (a), (b), (c). The
result for (a) has long been known of course.
For best answers the complexity is uniform: $\thetaptwo$-complete for
all variations (the result for (b) was previously known). We shall define these complexity classes in the next
section; for the reader not familiar with them, they all lie within
the second level of the polynomial hierarchy.

For the second theme of the paper, we look at query rewritings. This
is a standard way of leveraging database technology in the case of
incomplete or imperfect information, and such rewritings were heavily
used in data integration, data exchange, OBDA,
 query answering using views, consistent query answering etc.
\mycite{calvanese-et-al:vbqr,calvanese-et-al:vbqp,cali-et-al:rewrite-conquer,cali-et-al:qr-di}.
First-order rewritings are particularly useful, as they allow to use
the power of standard database query engines. In fact when they exist,
the rewritten queries can be implemented in any relational query
engine by expressing them in SQL, with no need to implement ad-hoc
algorithms. Next best are rewritings into Datalog (with negation):
these let us express queries using recursive features of SQL. 

As already mentioned, for unions of conjunctive queries (and even some
mild restrictions with guarded negation \mycite{GLS14}) certain
answers are computed by \naive\ evaluation {\em without} the presence of
constraints. 
Under constraints, even such simple ones as keys, the picture is less
complete. Indeed, keys, and in general equality-generating
dependencies (EGD) change the syntactic shape of a query that makes
naive evaluation work. 

\begin{itemize}
    \item Certain answers to a conjunctive query $Q$ (or a union of CQs) on a database $D$ under key constraints $\Sigma$ can be found by \naive\ evaluation of $Q$ on the result of the {\em chase} of $D$ with $\Sigma$. Mathematically, $\cert_\Sigma(Q,D)=Q(\chase_\Sigma(D))$, where on the left-hand side we have certain answers under constraints, and on the right hand side the \naive\ evaluation of $Q$ over the result of the chase. Here $\chase_\Sigma$ refers to the classical textbook chase procedure with keys, or more generally functional dependencies. In fact the above result applies when $\Sigma$ is a set of functional dependencies or equality generating dependencies (EGDs), not just keys.
\end{itemize}

Unfortunately the above result does not work when we move outside the
class of select-project-join-union queries, or unions of CQs. In fact
even without constraints, certain answers to a query of the form
$Q_1-Q_2$, where both $Q_1$ and $Q_2$ are CQs, are not necessarily
produced by \naive\ evaluation. To see why, take a database containing
one fact $R(1,\bot)$ where $\bot$ is a null and $Q_1$ returning $R$
while $Q_2$ is given by a formula $R(x,y) \wedge x=y$. Here \naive\
evaluation of $Q_1-Q_2$ returns $R$ (as $1$ is not equal to $\bot$), while certain answers is empty (as $0$ is a possible value for $\bot$).

This motivates our question whether we can extend the class of CQs and
their unions to obtain tractable evaluation of certain answers under
constraints such as functional dependencies and EGDs. The answer is
positive; in fact the query of the form $Q_1-Q_2$ above will be an
example of a query in this class. To start with, the class must be
such that finding certain answers for its queries without constraints
is already tractable. We know one such class: it consists of arbitrary
Boolean combinations of CQs, not just their union. We shall denote it
by BCCQ. It was proved in \mycite{gheerbrant-libkin:xml} that certain
answers for it can be found in polynomial time (for data complexity), though the procedure
was tableau-based and not suitable for implementation in
a database system.

This is precisely what we do in the second part of this paper. We
establish three main results:

\begin{enumerate}
    \item For an arbitrary BCCQ $Q$ and a set of EGDs $\Sigma$ one can
    construct a Datalog (with negation) query $Q'$ whose naive
    evaluation computes certain answers,  thereby ensuring their polynomial-time data complexity.
    \item There are however simple BCCQs, in fact even CQs, and keys,
    such that certain answers cannot  be expressed as a first-order
    queries. Therefore, using Datalog was necessary. 
    \item Without constraints present, certain answers to BCCQs are
    not only polynomial-time computable as had been shown previously,
    but also can  be expressed by first-order queries and thus
    efficiently implemented in SQL databases without using recursion. 
\end{enumerate}

The Sections \ref{dat}, \ref{nofo} and \ref{fo} address these items, respectively. 

Note that the material from this paper is based on the two conference papers \mycite{ijcai} and \mycite{datalog}.
%
%

\section{Preliminaries}
\label{prelim-sec}
\subsection*{Incomplete databases and constraints}
We represent missing information in relational databases in the standard way using nulls \mycite{abiteboul-et-al:dbbook,imielinski-et-al:incomp,vandermeyden:incomp-survey}. Incomplete databases are populated by \emph{constants} and \emph{nulls}, coming respectively from two countably infinite sets $\const$~and $\nulls$. We denote nulls by $\bot$, sometimes with sub- or superscript. We also allow them to repeat, thus adopting the model of \emph{marked} nulls, as customary in the context of applications such as OBDA or data integration and exchange.

A relational schema, or vocabulary $\sigma$, is a set of relation names with associated arities. A database $D$ over $\sigma$ associates to each relation name of arity $k$ in $\sigma$, a k-ary relation which is a finite subset of $(\const \cup \nulls)^k$. Sets of constants and nulls occurring in $D$ are denoted by $\const(D)$ and $\nulls(D)$. 
A database is complete if it contains no nulls, i.e. $\nulls(D) = \emptyset$.

The {\em active domain} of $D$ is the set of all values appearing in $D$, i.e. $\adom(D)=\const(D)\cup\nulls(D)$.

A \emph{valuation} $v : \nulls(D) \rightarrow \const$ on a database $D$ is a map that assigns constant values to nulls occurring in $D$. By $v(D)$ and $v(\bar a)$ we denote the result of replacing each null $\bot$ by $v(\bot)$ in a database $D$ or in a tuple $\bar a$.
The semantics $\sem D$ of an incomplete database $D$ is the set $\{v(D) \mid v \text{ is a  valuation on $D$}\}$ of all complete databases it can represent. Here as is common in research on incomplete data, we use closed world assumption \mycite{imielinski-et-al:incomp,reiter-cwa} (i.e., everything we don't know to be true is automatically assumed to be false and no new tuple can be added). 

An \emph{equality generating dependency} (EGD) is a first-order sentence of the form $\forall \bar x ~~(\varphi(\bar x)~\rightarrow~z=z')$, where $\varphi(\bar x)$ is a conjunction of atoms (without constants),  each variable in $\bar x$ occurs in some atom of $\varphi$, and $z, z'$ are distinct variables in $\bar x$.
As a special case, a \emph{functional dependency} (FD) over a relation name $R$ is of the form $\forall \bar x, \bar y, z, z'~~(R(\bar x, \bar y, z) \wedge R(\bar x, \bar y', z') \rightarrow z = z')$.
Throughout this paper we will assume that a (possibly empty) set of EGDs $\Sigma$ is associated with the database schema $\sigma$. 

A valuation $v$ is \emph{consistent} with $\Sigma$ (or just \emph{consistent}, when $\Sigma$ is clear from the context) if $v(D) \models \Sigma$. We denote by
$\vl(D)$ the set of all consistent valuations defined on $D$.

\subsection*{Query answering}

An $m$-ary {\em query} Q of active domain $\adom(Q) \subseteq \const$ is a map that associates with a database $D$ a subset of  $(\adom(D)\cup \adom(Q))^m$. 
To answer an $m$-ary query $Q$ over an incomplete database $D$ we follow \mycite{lipski} and adopt a slight generalisation of the usual intersection based certain answers notion, defined as $\cap_vQ(v(D))$, and furthermore incorporate constraints into query answering.

The set of \emph{certain answers} to $Q$ over $D$, with respect to a set of constraints $\Sigma$, is $$\cert_\Sigma(Q,D) = \{\bar{a} \in (\adom(D)\cup\adom(Q))^m\ \mid\ 
v(\bar{a})\in Q(v(D)) \textrm{ for all consistent } v\}\,.$$ For queries that explicitly use constants, we shall expand this to allow $\bar a$ range over $\adom(D)$ and those constants. The only difference with the usual notion is that we allow answers to contain nulls, to avoid pathological situations when answers known with certainty are not returned (e.g., in  a query returning a relation $R$ one would expect $R$ to be returned while the intersection-based certain answer will only return null-free tuples).  If the set of constraints $\Sigma$ is empty, we omit it and write simply $\cert(Q,D)$. Of course every valuation is consistent with the empty set of constraints. 

Following \mycite{libkin:zero}, given a query $Q$, a database $D$, a set of constraints $\Sigma$, and a tuple $\bar a$ 
over $\adom(D)\cup \adom(Q)$, 
we let the 
{\em support} of $\bar a$ be the set of
all valuations that witness it:
$$\supp_\Sigma(\qd,\bar a) \ = \ \{v\in \vl(D) \mid\ v(\bar a)\in Q(v(D))\}\,.$$
Again if $\Sigma=\emptyset$
 we omit the subscript. 

Supports thus measure how close a tuple is to certainty. We consider one answer to be
{\em better} than another if it has more support. That is, given a
database $D$, a $k$-ary query $Q$, and $k$-tuples $\bar a, \bar b$
over $\adom(D)\cup\adom(Q)$, we let
$$\begin{array}{rcl}
\bar a \sless_\qd^\Sigma \bar b & \LRA & \supp_\Sigma(\qda) \subset\supp_\Sigma(\qdb)\,.
\end{array}
$$
The set of {\em best answers} to $Q$ over $D$ is defined as the set of answers for which there is no better one: $$\best_\Sigma(Q,D) = \{\bar a \mid \neg\exists \bar b: \bar a
\sless_\qd^\Sigma \bar b\}.$$

As the set of {\em certain answers} to $Q$ over $D$ is the set of answers that are witnessed by all valuations, note that it could also be defined using the notion of support. Namely,
$\cert_\Sigma(Q,D)$ consists of all tuples $\bar{a} \in \adom(D)^m$ for which 
$\supp_\Sigma(\qda)$ contains all consistent valuations in $\vl(D)$.

\begin{example}\label{ex}
Let $Q(x)=\exists y (R(y) \wedge S(y,x))$ and $D = \{R(\bot_1), R(1), S(\bot_2, \bot_2)\}$. 

We have  $\supp(\qd,\bot_2) \ = \ \{v\in \vl(D) \mid\ v(\bot_2)=1\text{ or } v(\bot_1)=v(\bot_2)\}$, $\supp(\qd,1) \ = \ \{v\in \vl(D) \mid\ v(\bot_2)=1\}$ and $\supp(\qd,\bot_1) \ = \ \{v\in \vl(D) \mid\ v(\bot_1)=v(\bot_2)\}$. 

It follows that $\cert(Q,D) = \emptyset$ and $\best(Q,D) = \{(\bot_2)\}$.
\end{example}

\subsection*{\Naive\ evaluation and certain answers}

For a query $Q$ written in FO or Datalog, we write $Q(D)$ to mean that such a query is evaluated \naive ly. That is, if 
$D$ contains nulls, nulls of $D$ are treated as new constants in the domain of $D$, distinct from each other, and distinct from all the other constants in $D$ and $\varphi$.   
For example the query $\varphi(x,y) = \exists z~(R(x, z) \wedge R(z, y))$, on the database $D = \{R(1, \bot_1), R(\bot_1, \bot_2), R(\bot_3, 2)\}$ selects only the tuple $(1, \bot_2)$.

There are known connections between \naive\ evaluation and certain answers. If $\Sigma$ is empty and $Q$ is a union of conjunctive queries, then $\cert_\Sigma(Q,D)=Q(D)$, see \mycite{imielinski-et-al:incomp}. If $\Sigma$ contains a set of EGDs, then $\cert_\Sigma(Q,D)=Q\big(\chase_\Sigma(D)\big)$; cf.~\mycite{2012Greco}. Here $\chase_\Sigma$ refers to the standard chase procedure with a set of EGDs, see \mycite{abiteboul-et-al:dbbook}. 

\subsection*{Query languages}

Here we shall study best and certain answers to 
\emph{first-order} ($\FO$) queries, possibly in the presence of constraints, by means of their rewriting in $\FO$ and {\em Datalog}. \FO\ queries of vocabulary $\sigma$ use atomic relational and equality formulae and are closed under Boolean connectives $\wedge, \vee, \neg$ and quantifiers $\exists, \forall$. 
We write  $\varphi(\bar x)$ for an $\FO$-formula $\varphi$ with free variables $\bar x$. 
With slight abuse of notation, $\bar x$ will denote both a tuple of variables and the set of variables occurring in it.
The set of constants used by $\varphi$ is denoted by $\adom(\varphi)$.
We interpret $\FO$-formulas under active domain semantics, i.e. quantified variables range over $\adom(D) \cup \adom(\varphi)$.
Thus, an $\FO$ formula $\varphi(\bar x)$ represents a query (of active domain $\adom(\varphi)$) mapping each database $D$ into the set of tuples $\{\bar t \textrm { over } \adom(D)\cup\adom(\varphi)~|~D \models \varphi(\bar t)\}$.

To evaluate $\FO$-formulas with free variables we use assignments $\nu$ from variables to constants in the active domain.
Note that with a little abuse of notation we write $D \models \varphi(\bar t)$ for $D \models_\nu \varphi(\bar x)$ under the assignment $\nu$ sending $\bar x$ to $\bar t$.

Here it is important to note that the query associated to $\varphi$ is a mapping defined on all databases $D$, possibly with nulls. If $D$ contains nulls, $D \models \varphi(\bar t)$ is to be intended ``naïvely", i.e. nulls of $D$ are treated as new constants in the domain of $D$, distinct from each other, and distinct from all the other constants in $D$ and $\varphi$.   
For example the query $\varphi(x,y) = \exists z~(R(x, z) \wedge R(z, y))$, on the database $D = \{R(1, \bot_1), R(\bot_1, \bot_2), R(\bot_3, 2)\}$ selects only the tuple $(1, \bot_2)$.

{\em Conjunctive queries} (CQs) are given by the $\exists,\wedge$-fragment of \FO, and their unions (UCQs) by the  $\exists,\wedge,\vee$-fragment of $\FO$; these are also captured by the positive fragment of relational algebra (select-project-union-join queries).

We also consider {\em Boolean combination of conjunctive queries (BCCQs)}, i.e., the closure of conjunctive queries under operations $q \cap q'$, $q \cup q'$, and $q - q'$.



A \emph{Datalog rule} \mycite{abiteboul-et-al:dbbook} is an expression of the form
$R_1(u_1) \leftarrow R_2(u_2), \ldots , R_n (u_n )$
where $n \geq 1$, $R_1, \ldots , R_n$ are relation names and $u_1,\ldots, u_n$ are tuples of appropriate
arities. Each variable occurring in $u_1$ must occur in at least one of $u_2, \ldots , u_n$. A \emph{Datalog
program} is a finite set of Datalog rules.
The \emph{head} of the rule is the expression $R_1(u_1)$; and $R_2(u_2), \ldots, R_n(u_n )$ forms the body. The semantics is the standard fixed-point semantics.

As the language of our rewritings, we shall be using $\FO$, but also a fragment of \emph{stratified Datalog with negation} in bodies that can be seen in two different ways. 
\begin{enumerate}
    \item A program is evaluated in two steps. First, we can have a Datalog program $P$ defining new idb predicates $S_1,\ldots,S_\ell$. Then we ask an \FO\ query over the schema extended with these predicates $S_1,\ldots,S_\ell$. 
    \item We evaluate a stratified Datalog with negation program in which the first stratum has no negation (but may have recursion) and the second stratum has no recursion (but may have negation). 
\end{enumerate}

From the rewritings we produce it will be clear that they fall in these classes. The key point about them is that they can be implemented in recursive SQL, and that they both have PTIME data complexity, making their evaluation feasible. Note that recursive SQL as it is currently implemented, e.g., in PostgreSQL 8.4, is actually Turing complete \mycite{CyclicTag,TM}.

\subsection*{Complexity classes}
In order to study the complexity of best and certain answer computation we shall need two classes in the second level of the
polynomial hierarchy. 
Both of these contain \np\ and \conp, and are  contained in
$\Sigma^p_2\cap\Pi^p_2$.
The class \classdp\ consists of languages $L_1\cap L_2$ where
$L_1\in\mbox{\np}$ and $L_2\in\mbox{\conp}$. 
This class has appeared in database applications \mycite{fagin-et-al:de,barcelo:approx}. 
The class $\thetaptwo$ consists of problems that can be solved in
polynomial time with a logarithmic number 
of calls to an
\np\ oracle \mycite{PNPlogn}. Equivalently, it can be described as the class of
problems solved in \ptime\ with an \np\ oracle where calls to the
oracle are done in parallel, i.e., independent of each other. 
This class has appeared in the context of AI, modal logic, OBDA \mycite{gottlob:carnap,eiter-et-al:fct,calvanese-et-al:obda,bienvenu-et-al:obda}, data exchange \mycite{arenas-et-al:beyond}.

\section{Complexity of Best and Certain Answers}
\label{complexity-sec}

%
%

We start by looking at complexity of certain and best answers of first-order queries, and answer a few questions that are (perhaps somewhat surprisingly) missing in the literature. 
In this case we look at arbitrary first-order queries; thus we do not mention constraints since $\cert_\Sigma(Q,D)=\cert(\Sigma\to Q,D)$ and likewise for best answers. In the subsequent sections, when we consider rewritings for sublanguages of first-order, we shall again mention constraints explicitly since queries of the form $\Sigma\to Q$ will normally {\em not} belong to the same syntactic class as $Q$ itself. In this context, we will refer to $\xanswer_\Sigma(Q)$.

As is common in database theory, we look at complexity in terms of complexity classes, which necessitates looking at {\em decision versions} of problems. 
The most common one that is found, stated here for 
$\x \in \{${\sc Certain}, {\sc Best}$\}$, is the following problem:
 


\begin{center}
\framebox(250,70){
\begin{tabular}{ll}
{\sc Problem}: & $\xanswer(Q)$\\
{\sc Input}: & A  database $D$, 
              a tuple $\bar a$\\
{\sc Question}: & Is $\bar a \in \x(Q,D)$?
\end{tabular}
}
\end{center}

We are thus interested in {\em data complexity}: the query is fixed. We do not study combined complexity in this paper. In the remainder, we thus often omit the query and write $\xanswer$ instead of $\xanswer(Q)$. Recall that in this case, for a language $L$, we say that the problem $\xanswer$ is $C$-{\em complete in data complexity} for a complexity class $C$, if $\xanswer(Q)$ is solvable in $C$ for every $Q\in L$, and there exists a specific $Q_0\in L$ so that $\xanswer(Q_0)$ is hard for $C$. We know from \cite{abiteboul-et-al:ca} that $\textsc{CertainAnswer}(Q)$ is \conp-complete in data complexity for first-order queries. 

For best answers, it is a different version of the decision problem for which the complexity is known. Specifically, 
\mycite{libkin:zero} considered the problem of checking whether the set $\x(Q,D)$ belongs to a specified family of sets:

\begin{center}
\framebox(270,70){
\begin{tabular}{ll}
{\sc Problem}: & $\familyxanswer(Q)$\\
{\sc Input}: & A database $D$, 
              a family $\XX$ of sets of tuples\\
{\sc Question}: & Is $\x(Q,D) \in \XX$?
\end{tabular}
}
\end{center}

For this decision version, the complexity of the problem was shown to be $\thetaptwo$-complete. This version looks a bit artificial, but we include it for the sake of completeness, because it has
appeared in the literature. 

However this presentation of a decision suggests another rather natural presentation of a decision version, namely asking 
if a given set is $\x(Q,D)$:

\begin{center}
\framebox(250,70){
\begin{tabular}{ll}
{\sc Problem}: &  $\setxanswer(Q)$\\
{\sc Input}: & A database $D$, 
    a set $X$ of tuples \\
{\sc Question}: & Is $\x(Q,D) = X$?
\end{tabular}
}
\end{center}

Our current state of knowledge is the complexity of {\sc CertainAnswer} (\conp-complete) and 
\familybestanswer ($\thetaptwo$-complete). Thus we now fill the gap and classify complexities of all the problems -- for data complexity -- in the case of FO queries.

We start by showing that all the alternatives for best answers -- \bestanswer, \setbestanswer, and \familybestanswer -- are computationally equivalent. 


\begin{theorem}\label{thm:thetaptwo}
For $\FO$ queries the problems \bestanswer, \familybestanswer~and \setbestanswer~ are $\thetaptwo$-complete in data complexity.
\end{theorem}

\begin{proof}
The upper bound for \setbestanswer~immediately follows from the upper bound for \familybestanswer~
(take the family $\XX$ to be a singleton $\{X\}$). 
As for \bestanswer~we only need a slight modification of the upper bound proof in \mycite{libkin:zero}.  
To check whether $\bar a \in \best(Q,D)$ we proceed as follows. 
Since the query is
fixed, and has therefore fixed arity $k$, in polynomial time we can enumerate all the
$k$-tuples of $\adom(D)$. Then, using parallel calls to the \np\ oracle, we can check for each
such tuple $\bar b$ whether $\supp(\qda) \subseteq\supp(\qdb)$ and whether $\supp(\qdb) \subseteq\supp(\qda)$. With this information, in polynomial time we know
whether $\bar a \sless_\qd \bar b$ for some $\bar b$.

Assuming $\Sigma$ empty, we prove the two remaining lower bounds, reducing from the same $\thetaptwo$-complete problem \mycite{Wagner:bqc}: given an undirected graph
$G$, is its chromatic number $\chi(G)$ odd? With each undirected graph $G=\langle N, E\rangle$ with nodes
$N$ and edges $E$, we associate a database $D_G$ over binary relations $L,E$ and
unary relations $C,O$ as follows. We use a null $\bot_n$ in $D_G$ for each node $n$ of $G$. For each edge $\{n,n'\}$ of $G$, we have pairs $(\bot_n,\bot_{n'})$ and 
$(\bot_{n'},\bot_n)$ in the relation $E$ of $D_G$. In relation $C$ we have $m$ constants $\{c_1,\ldots,c_m\}$ (intuitively representing possible colors), where $m$ is the number of nodes of $G$.
Relation $O$ of $D_G$ is defined as $\{c_i ~|~ i ~ \textrm{is odd}\}$ and $L$ is a linear ordering on them, i.e., $(c_i,c_j)\in L$ iff
$i\leq j$, for $i,j\leq m$.

Remark that any valuation $v$ of $D_G$ that maps each null into a constant of $C$ represents an assignment of colours in $\{c_1,\ldots,c_m\}$ to nodes of $G$.
Then we define a query 
$$\phi(x) \ \ =\ \  
C(x) \wedge  \forall y,z\ \big(E(y,z) \to L(y,x)\big) \wedge \forall y\ \big(L(y,x) \to \exists z\ E(y,z)\big)\wedge \neg \exists y\ E(y,y).$$
For any valuation $v$, $\phi(c)$ holds in $v(D_G)$ iff
1) $c = c_j$ for some $j=1..m$ (ensured by the first conjunct). 
2) For such a $c_j$, the valuation $v$ maps each null into $\{c_1,\ldots,c_j\}$ (second conjunct), 
i.e. $v$ represents an assignment of colours to nodes of $G$, using at most the first j colours.
3) Each color $\{c_1,\ldots,c_j\}$ is used by $v$, i.e. $v$ represents an assignment of colours 
to nodes of $G$, using precisely the first j colours (third conjunct).
4) There are no loops in $E$ (fourth conjunct).

Thus, for a valuation $v$, the 
formula $\phi(c_j)$ is true in $v(D_G)$ iff $v$ represents a colouring of G using precisely the first j colours 
$\{c_1,\ldots,c_j\}$ (which in the sequel we refer to as an \emph{exact j-colouring} of $G$).

Next we define:
$$Q(x) = C(x) \wedge  (\phi(x) \vee \exists y~( O(y) \wedge L(x, y) \wedge \phi(y)))$$

For a valuation $v$, we have that $Q(c_i)$ holds in $v(D_G)$ iff 
either $v$ represents an exact $i$-coloring of $G$;
or $v$ represents an exact $j$-coloring of $G$ with $j$ odd, and $i \leq j$.
In other words valuations representing exact $j$-colorings, with $j$ even, support only the maximal color $c_j$; while 
valuations representing exact $j$-colorings, with $j$ odd, support all colors $\{c_1...c_j\}$.

With this in place we can conclude the reduction for the \bestanswer~problem:
\smallskip

\textbf{Claim.} $c_1 \in \best(Q,D_G)$ iff the chromatic number of $G$ is odd.


First, we prove the above claim. 
Let $\chi_G$ be the chromatic number of $G$. Then there exist no exact colorings of $G$ which are prefixes of 
$\{c_1, \dots c_{\chi_G}\}$, while $\{c_1, \dots c_{\chi_G}\}$ is an exact coloring of $G$.

Assume first that $\chi_G$ is even.  Then there exist no valuations representing the exact coloring $\{c_1\}$. Thus the support of $c_1$ is the set of valuation representing an exact coloring $\{c_1...c_j\}$ of $G$ with $j$ odd and $j > \chi_G$.
This support is not maximal, In fact the support of $c_{\chi_G}$ is:
\begin{itemize}
\item the valuations representing the exact coloring $\{c_1...c_{\chi_G}\}$ (there exists at least one);
\item the valuations representing an exact coloring $\{c_1...c_j\}$ of $G$ with $j$ odd and $j > \chi_G$.
\end{itemize}

This support  strictly contains the support of $c_1$; in fact valuations in the first item cannot be also in the second.

Assume now that $\chi_G$ is odd. Then the support of $c_1$ is the set of valuations representing an exact coloring $\{c_1...c_j\}$ of $G$ with j odd and $j \geq \chi_G$. We show that this set is maximal, i.e. no color $c_k$ can have a support strictly containing it.

\begin{itemize}
\item if $k \leq \chi_G$ then the support of $c_k$ is the set of valuations representing an exact coloring $\{c_1...c_j\}$ of $G$ with j odd, and $j \geq \chi_G$. So same support as $c_1$.
\item if $k > \chi_G$, the support of $c_k$ cannot contain the valuations representing $\{c_1, \dots c_{\chi_G}\}$. There exists at least one such valuation and it belongs to the support of $c_1$. Thus the support of $c_k$ does not contain the support of $c_1$.
\end{itemize}

We now move to \setbestanswer. With any undirected graph $G$ we associate a relational structure $D'_G$ obtained from $D_G$ by adding a new colour $c_0$ in $C$ with $L(c_0,c_i)$ for every $0 \leq i \leq m$. 
We define a restriction $\psi$ of the original formula $\phi$ by disallowing $c_0$ in colourings: to obtain $\psi$ it suffices to replace $L(y,x)$ in $\phi$ by $L(y,x) \wedge y \neq c_0$, and $C(x)$ by  $C(x) \wedge x\neq c_0$.
Thus it is still true that $\psi(c_j)$ is true in $v(D_G')$ iff $v$ represents a colouring of G using precisely 
$\{c_1,\ldots,c_j\}$.

We define a new query:

\vspace{2mm}
\hspace{-2mm}
\begin{tabular}{r l}
$Q'(x) :=O(x) \wedge$ & \!\!\!\!\!\!$(\psi(x) \vee \exists y~( O(y) \wedge L(x, y) \wedge \psi(y))$  \\
                      & $\vee$ \\
                      $\!\!\!\!\!\neg O(x) \wedge$ & \!\!\!\!\!\!$(\psi(x) \vee \exists y~( O(y) \wedge x+2 < y \wedge \psi(y))$ \\
                      & $\vee$ \\
                      $\!\!\!\!\!\neg O(x) \wedge$ & \!\!\!\!\!\!$\exists y (x\neq y \wedge L(x,y)) ~\wedge~ \forall y \forall z~( E(y,z)\rightarrow (y=c_0 \wedge z=c_0))$ \\
\end{tabular}

\vspace{2mm}

Note that $x+2 < y$ is used as a shorthand, as it is definable in our language.

\vspace{2mm}

$Q'(c_i)$ holds in $v(D'_G$) iff 

\begin{itemize}
\item either $i$ is odd and $v$ represents an exact $j$-colouring of $G$, with $j$ odd and $i \leq j$;
\item or $i$ is even and:
\begin{itemize}
\item either $v$ represents an exact colouring $\{c_1...c_j\}$ of $G$ with j odd, and $i+2 < j$;
\item or $v$ represents an exact colouring $\{c_1...c_i\}$ of $G$;
\item or $i<m$ and $v(\bot_j)=c_0$ for all $1 \leq j \leq m$;
\end{itemize}
\end{itemize}

The following claim completes the reduction for \setbestanswer~:
\smallskip


\noindent
\textbf{Claim.} $\{c_i ~|~ i \textrm{ is even}\}=\best(Q',D'_G)$ iff $\chi(G)$ is even.

In the following, we call $v_0$ the unique valuation such that $v_0(\bot_j)=c_0$ for all $1 \leq j \leq m$. 
First assume that $\chi_G$ is even. For all $0 < i \leq m$ odd, $Supp(c_i)$ is not maximal as $Supp(c_0) \supset Supp(c_i) \cup \{v_0\}$. Hence $\best(Q',D'_G) \subseteq \{c_i ~|~ i \textrm{ is even}\}$, so
we show $\{c_i ~|~ i \textrm{ is even}\} \subseteq \best(Q',D'_G)$. 
The inclusion holds whenever $c_i \geq \chi(G)$, as $Supp(c_i)$
contains all valuations representing exact colorings $\{c_1...c_i\}$ of G, while no other
$Supp(c_j)$ with $i \neq j$ contains them. 
Now take $c_i < \chi(G)$ with $i$ even, then $Supp(c_i)$ contains $v_0$ together with all exact odd colourings 
(if there are any). First assume that there exists odd exact colourings of $G$, so there are $\chi(G)+1$ ones and valuations representing them are not contained in $Supp(\chi(G))$.
Also $v_0 \not \in Supp(c_k)$ with $k$ odd and $k < \chi(G)$. It follows that $Supp(c_i)$, which is the union of $v_0$ and of all valuations representing odd exact colorings is maximal. 
Now assume that there is no exact odd colouring. This corresponds to the special case $\chi(G)=m$ where $Supp(c_m)$ contains only the exact colourings $\{c_1...c_{m}\}$ of $G$, but not $v_0$ ; while $Supp(c_j)=\emptyset$ whenever $j$ odd. In such a case $Supp(c_i)=\{v_0\}$ is also maximal.

We assume now $\chi(G)$ is odd and show $\{c_i ~|~ i \textrm{ is even}\} \neq \best(Q',D'_G)$. 
First notice that $Supp(c_1)$ is maximal whenever $\chi(G)=1$, as neither $Supp(c_0)$, nor any $Supp(c_i)$ with $i \geq 2$ contain valuations representing the exact $\{c_1\}$ colourings. So we assume $\chi(G)\geq 3$, from which it follows that there exists a constant $c_{\chi(G)-3}$ in the active domain which support contains $v_0$ together with all valuations representing exact odd colourings. As $Supp(c_{\chi(G)-1})$ contains exactly the same set of valuations, to the exclusion of those representing $\{c_1...c_{\chi(G)}\}$ colourings, it follows that $Supp(c_{\chi(G)-3}) \supset Supp(c_{\chi(G)-1})$ and so $c_{\chi(g)-1} \not \in \best(Q',D'_G)$.
\end{proof}

Now that we showed that all three formulations of best answers actually collapse computationally, another natural question arises. Does a similar result hold for certain answers ? 
It is well known that data complexity of \certainanswer~ is \conp-complete for \FO\--queries \mycite{abiteboul-et-al:ca}. We complete the picture as follows.
\begin{theorem}
For $\FO$ queries, \setcertainanswer~is $\classdp$-complete and \familycertainanswer~is $\thetaptwo$-complete in data complexity.
\end{theorem}

\begin{proof}
To prove membership of \setcertainanswer~in $\classdp$, notice that for a query $Q$, this problem is the intersection of two languages $L_1 \cap L_2$
where $L_1=\{(D,X)~|~X \subseteq \cert(Q,D)\}$ and $L_2=\{(D,X)~|~\overline{X} \subseteq \overline{\cert(Q,D)}\}$.
$L_1$ is known to be in \conp~: we guess a tuple $\bar a \in X$ and a valuation $v\in V(D)$ with $v(\bar a)\not \in Q(v(D))$. Similarly, $L_2$ is in \np~: we guess a tuple $\bar b \in \overline{X}$ and a valuation $v'\in V(D)$ with $v'(\bar b) \in Q(v(D))$.

To prove membership of \familycertainanswer~in $\thetaptwo$, 
suppose the query $Q$ is k-ary, and we are given a family of sets of k-ary tuples $\XX=\{X_1, \ldots, X_n\}$ and a database $D$. For each $X_i \in \XX$, we use the \np~oracle to decide in parallel whether $X_i=\cert(Q,D)$ (for each $X_i$ the two calls to the oracle do not depend on each other and they can also be done in parallel).

For $\classdp$-hardness, we reduce from the problem of checking whether $\chi(G)$, the chromatic number of an undirected graph $G$, equals 4 \mycite{rothe:dp} and for $\thetaptwo$~-hardness, we reduce from the related problem of checking whether $\chi(G)$ is odd. With such a graph $G$, we
associate the same database $D_G$ as in the proof of Theorem \ref{thm:thetaptwo}. Using the exact coloring formula $\varphi$ in the proof of Theorem \ref{thm:thetaptwo}, we define a query 

\begin{align}
Q(x):=C(x) \wedge \forall y~ (\varphi(y) \rightarrow L(x,y))\nonumber
\end{align}
We claim that $\cert(Q,D)=\{c_1, \ldots, c_n\}$ iff $\chi(G)=n$, which entails $\cert(Q,D)=\{c_1, \ldots, c_4\}$ iff $\chi(G)=4$ and $\cert(Q,D) \in \{\{c_1,\ldots,c_j\}~|~ j \text{ is odd and }1 \leq j \leq |G|\}$ iff $\chi(G)$ is odd.
Recall that $v(D_G)\models \varphi(c_i)$ iff $c_i$ is a color in $\{c_1, ..., c_{|G|}\}$ and $v$ represents an exact $i$-coloring of the graph. Now $v(D_G)\models Q(c_j)$ iff $c_j$ is a color and there is no  $i < j$ such that $v$ represents an exact $i$-coloring of the graph, which holds exactly whenever $c_j \in \{c_1, ..., c_{\chi(G)}\}$.
\end{proof}

\begin{figure*}
\begin{center}
\bgroup
\def\arraystretch{1.7}
\begin{tabular}{|l|c|c|c|}
\cline{1-4}
Type of problem & $\bar a \in$ Answer & $X=$ Answer & Answer $\in
{\mathcal X}$ \\ 
\cline{1-4}
Certain Answer & coNP-complete \footnotemark[1] & \classdp\-complete & $\text{\rm
   P}^{\text{\rm NP}[\log n]}$-complete \\ 
\cline{1-4}
Best Answer &
$\text{\rm  P}^{\text{\rm NP}[\log n]}$-complete &
$\text{\rm  P}^{\text{\rm NP}[\log n]}$-complete &
$\text{\rm  P}^{\text{\rm NP}[\log n]}$-complete \footnotemark[2] \\ 
\cline{1-4}
\end{tabular}
\egroup
\end{center}
 \caption{Summary of data complexity results for \FO\ queries}\label{compl-fig}
\end{figure*}

\footnotetext[1]{\mycite{abiteboul-et-al:ca}}
\footnotetext[2]{\mycite{libkin:zero}}
%
%


\section{Query Rewritings for Tractable Fragments}
\label{rewriting-sec}

Considering arbitrary \FO-queries brought us an intrinsic intractability result for all variants of the considered decision problems. This motivates restricting to well behaved fragments such as CQs and UCQs. 
Recall that {\em conjunctive queries} (CQs) are given by the $\exists,\wedge$-fragment of \FO, and their unions (UCQs) by the $\exists,\wedge,\vee$-fragment of $\FO$.
We extend them with a mild form of negation (since adding negation leads to coNP-hardness of certain answers). 
This mild form comes in the shape of 
{\em Boolean combination of conjunctive queries (BCCQs)}, i.e., the closure of conjunctive queries under operations $q \cap q'$, $q \cup q'$, and $q - q'$.

If there are no constraints in $\Sigma$, finding certain answers to BCCQs is known to be 
tractable \mycite{gheerbrant-libkin:xml}, though by tableau-based techniques that are hard to implement in a database system.
We now extend this in two ways. First, we show that tractability is preserved even in the presence of EGDs (and thus functional dependencies and keys). Second, we show that certain answers can be obtained by rewriting into a fragment of Datalog as described in Section \ref{prelim-sec}. In particular, it means that certain answers can be found by a query expressible in recursive SQL (and even in SQL in the absence of constraints).

For \bestanswer\, a polynomial time evaluation algorithm (in data complexity) already exists \cite{libkin:zero}. The resolution based procedure is however in sharp contrast with 
 na\"ive evaluation, which allows to compute certain answers to unions of conjunctive queries via usual model checking. We thus show how to apply our query rewriting techniques to the best answers problem. 

\subsection{A normal form for queries : neutralizing variable repetition}
Towards our rewritings, we start by putting each conjunctive query in a normal form which eliminates repetition of variables, by introducing new equality atoms. 

\begin{definition}[NRV normal form]
A conjunctive query $Q$ is in {\em non-repeating variable normal form} (NRV normal form) whenever it is of the form 
$Q(\bar{x})~=~\exists \bar{w} ~ (q(\bar{w}) \wedge e(\bar{x}, \bar{w}))$
where variables in $\bar x\bar w$ are pairwise distinct, and:
\begin{itemize}
\item $q(\bar{w})$
is a conjunction of relational atoms without constants, where each free variable in $\bar{w}$ has at most one occurrence in $q$,
\item $e(\bar{x},\bar{w})$ is a conjunction of equality atoms, possibly using constants, where each variable of $\bar x$ is involved in at least one equality.
\end{itemize}
We say that $q(\bar{w})$ is the \emph{relational subquery} of $Q$, and $e(\bar{x}, \bar{w})$ is the \emph{equality subquery} of $Q$.
A BCCQ is in NRV normal form if it is a Boolean combination of CQs in NRV normal form.
\end{definition}


\begin{example}\label{ex::nrv}
The query $Q(x)$ from Example \ref{ex} is equivalent to $\exists w_1 w_2 w_3 (R(w_1) \wedge S(w_2,w_3) \wedge w_1=w_2 \wedge w_3=x)$, which is in NRV normal form.
\end{example}

Clearly every conjunctive query $Q$ is equivalent to a query in NRV normal form; moreover $Q$ can be easily rewritten in NRV normal form (in linear time in the size of the query). Thus in what follows we assume w.l.o.g. that conjunctive queries are given in NRV normal form.
Intuitively the NRV normal form allows us to separate the two ingredients of a conjunctive query : the existence of facts in some relations of the database on the one side, and a set of equality conditions on data values occurring in these facts, on the other side.
The existence of facts does not depend on the valuation of nulls, and thus can be directly tested on the incomplete database.
Instead equality atoms in an NRV normal form imply conditions that valuations need to satisfy in order for the query to hold.
We can thus first concentrate on the support of
equality subqueries. This will be encoded in \FO~ and then integrated in the rewriting of the whole conjunctive query.

We introduce a notion of equivalence of database elements w.r.t. to a set of equalities. Intuitively equivalent elements of a tuple $\bar t$ are the ones which should be collapsed into a single value in order for a valuation of $\bar t$ to satisfy all the given  equalities.

\begin{definition}
Given a database $D$, a conjunction of equality atoms $\gamma(\bar{y})$ and an assignment $\nu: \bar y \cup \adom(\gamma) \rightarrow\adom(D)\cup \adom(\gamma)$ preserving constants,
we say that $u,u' \in \adom(D)\cup\adom(\gamma)$  are equivalent w.r.t. $\gamma$ and $\nu$ and write $u \equiv_\gamma^\nu u'$, 
if either $u=u'$ or $(u, u')$ belongs to the reflexive symmetric transitive closure of $\{(\nu(x), \nu(w)) ~|~ x = w \in \gamma\}$.
\end{definition}

The relation $\equiv_\gamma^\nu$ is clearly an equivalence relation over $\adom(D) \cup\adom(\gamma)$, where each element outside the range of $\nu$ forms a singleton equivalence class.

\begin{example}\label{example::equiv}
Let  $\gamma$ be $x_1=x_2 \wedge x_2=x_3 \wedge x_4=x_5 \wedge x_6=1$. Let $\nu$ assign $\bot_i$ to $x_i$ for $i \leq 5$, and $\bot_5$ to $x_6$. The equivalence classes of  $\equiv_\gamma^\nu$ are $\{\bot_i~|~i \leq 3\}$ and $\{1, \bot_4, \bot_5\}$, plus one singleton for each other element of the active domain.
\end{example}

In what follows we denote by $\sim_{\gamma}$ the reflexive symmetric transitive closure of $\{ (x,w) ~|~ x = w \in \gamma \}$.
Note that this is an equivalence relation among variables and constants of $\gamma$. We will provide two syntactic encodings of this relation, one in Datalog and one in FO.
\label{normal_form}

\subsection{Datalog Rewriting for Certain Answers for BCCQs with EGDs}
%
%

Recall that, given a query $Q$, a database $D$, and a tuple $\bar{a}$ over $\adom(D) \cup \adom(Q)$ we let the support of $\bar{a}$ be the set of all valuations that witness it :
$$\supp(Q, D, \bar a)=\{v \in \vl(D)~|~v(\bar a)\in Q(v(D))\}$$

In order to look for rewritings of BCCQs, a key observation is that 
$\bar a$ is a certain answer to $Q$ iff $\supp(\neg Q, D, \bar a) = \emptyset$. 
When $Q$ is a BCCQ, so is $\neg Q$, thus 
we look for ways of expressing (non-)emptiness of the support for BCCQs.

We start by concentrating on the support of
equality subqueries. This will be encoded in Datalog and then integrated, as a key ingredient, in the rewriting of the whole  query. 
We let $\gamma(\bar y)$ be an arbitrary set of equality atoms among variables $\bar y$ and possibly constants. Intuitively we will be interested in the case that $\gamma (\bar y)$ is the equality subquery $e(\bar x, \bar w)$ of a CQ in NRV normal form (thus notice that in the Datalog program below $\bar y$ encompasses variables $\bar x \bar w $ of an equality subquery).

Remark that we can always write an EGD so that no variable in its body occurs more than once; it suffices to add to the body a set of variable equalities. 
Thus we assume that EGDs in $\Sigma$ are of the form  $\forall \bar u ((\varphi(\bar{u}) \wedge \psi) \rightarrow z=z')$ where $z, z'$ are in $\bar u$, the conjunction of atoms $\varphi(\bar u)$ contains no constants, no variable occurs twice in $\varphi(\bar u)$, and $\psi$ is a set of equalities among variables of $\bar u$.
Remark also that membership in the set $\adom(D) \cup\adom(\gamma)$ can be expressed by a UCQ formula that we call $Dom(x)$. We encode equivalence of database elements in $\adom(D) \cup\adom(\gamma)$ w.r.t. a set of equalities $\gamma(\bar y)$ 
using the following 
Datalog program 
\footnote{Queries we write hereafter can be domain dependent. So it is important to recall that we always use active domain semantics. 
}:


\begin{align}
& equiv_\gamma(\bar{y}, z, z) \leftarrow  \wedge_i~ Dom(y_i), Dom(z) \nonumber \\
& equiv_\gamma(\bar{y}, z, z') \leftarrow  z=y_k, z'=y_l, \wedge_i~ Dom(y_i) ~ \textrm{  for each } (y_k=y_l) \in \gamma   \nonumber \nonumber \\ 
& equiv_\gamma(\bar{y}, z, z') \leftarrow  equiv_\gamma(\bar{y}, z, u), equiv_\gamma(\bar{y}, u, z') 
\nonumber \nonumber\\
& equiv_\gamma(\bar{y}, z, z') \leftarrow  equiv_\gamma(\bar{y}, z', z) \nonumber \nonumber \\ 
& equiv_\gamma(\bar{y}, z, z') \leftarrow 
\varphi(\bar{u})~\wedge_{(w = w')\in \psi(\bar u)} equiv_\gamma(\bar{y}, w, w') \nonumber \\
& ~~~~~~~~~~~~  \textrm {for each EGD } \forall \bar u((\varphi(\bar{u})\wedge \psi(\bar u)) \rightarrow z=z') \in \Sigma \nonumber
\end{align}

 Intuitively, if $\bar t$ is a tuple of database elements assigned to $\bar y$, equivalent elements of $D$ are the ones which should be collapsed into a single value in order for a valuation of $D$ to satisfy all the  equalities $\gamma (\bar t)$ and the EGDs.
For fixed $\gamma$ and $\bar t$, the relation $\{(s,s') ~|~ D \models equiv_{\gamma}(\bar t, s, s')\}$ is an equivalence relation over $\adom(D) \cup \adom(\gamma)$ where each element of $\adom(D)$ neither in $\bar t$ nor in $\adom(\gamma)$ forms a singleton equivalence class. 

The formula $equiv_\gamma$ is a key ingredient in our rewriting;  as formalized in the following lemma, it selects precisely the pairs of elements that a consistent valuation needs to collapse to satisfy a set of equalities.


\begin{lemma}\label{intermediate}
Let $\gamma(\bar y)$ be a conjunction of equality atoms, $D$ a database, 
and $\nu(\bar{y})=\bar{t}$ an assignment over $\adom(D) \cup \adom(\gamma)$. Assume $v$ is a consistent valuation of nulls,
 then $v(D) \models \gamma(v(\bar{t}))$ 
if and only if $v(s)=v(s')$ for all $s, s'$ such that
$D \models equiv_{\gamma}(\bar{t},s,s')$. 
\end{lemma}
\begin{proof}
$\Rightarrow$ Assume $v(D) \models \gamma(v(\bar{t}))$ and let $s, s'$ 
such that $D \models equiv_{\gamma}(\bar{t},s,s')$. We prove $v(s) = v(s')$.
We proceed by induction on the derivation of $equiv_{\gamma}(\bar{t},s,s')$ by the fixpoint 
evaluation of the Datalog program. Assume $equiv_{\gamma}(\bar{t},s,s')$ is derived at the first iteration, 
then it follows from one of the first two rules. If it is derived by the fist rule then $s = s'$ and 
therefore $v(s) = v(s')$ trivially. Assume $equiv_{\gamma}(\bar{t},s,s')$ is derived using the second 
rule then there exists $(y_k = y_l) \in \gamma$, and $s= t_k$ and $s'=t_l$; now since $v(D) \models \gamma(v(\bar{t}))$, we have $v(t_k) = v(t_l)$.
Now assume that $equiv_{\gamma}(\bar{t},s,s')$ is derived at some subsequent step. 
If it follows from the second rule then it follows from $equiv_{\gamma}(\bar{t},s,p)$ and $equiv_{\gamma}(\bar{t},p,s')$ derived at previous steps, thus by the induction hypothesis $v(s) = v(p)=v(s')$. Similarly if $equiv_{\gamma}(\bar{t},s,s')$ is derived by the third rule, then it follows from $equiv_{\gamma}(\bar{t},s',s)$, and thus by induction $v(s) = v(s')$.
The last case is that $equiv_{\gamma}(\bar{t},s,s')$ is derived by the last rule for some 
EGD $(\varphi(\bar{u})\wedge \psi(\bar u)) \rightarrow z=z')$. 
So there exists a mapping $\mu$ of $\bar u$ such that $D \models \phi(\mu(\bar u))$ and 
$D\models equiv_\gamma(\bar t, \mu(w), \mu(w'))$ for each $(w=w') \in \psi(\bar u))$ and 
$s =\mu(z)$ and $s'=\mu(z')$; so by the induction hypothesis $v(\mu(w)) = v(\mu(w'))$. 
It follows that $v(D) \models \phi(v(\mu(\bar u)) ) \wedge \psi(v(\mu(\bar u)))$, 
and because $v(D)$ satisfies the EGDs ($v$ being consistent), $v(\mu(z)) = v(\mu(z'))$, thus $v(s) = v(s')$.


$\Leftarrow$ Assume $\forall s, s'$, 
$D \models equiv_{\gamma}(\bar{t},s,s')$ implies 
$v(s)=v(s')$.
We show that 
$v(D) \models \gamma(v(\bar{t}))$. 
In fact for each $(y_k = y_l) \in \gamma$, we have $D \models equiv_\gamma(\bar t, t_k, t_l)$ (derived by the second rule).
By our hypothesis $v(t_k) = v(t_l)$, thus $\gamma(v(\bar t))$ holds.
\end{proof}

\begin{example}
Let $\gamma$ and $\nu$ be as in Example~\ref{example::equiv}, then Lemma~\ref{intermediate} implies that a valuation $v(D) \models \gamma(v(\bar{t}))$ iff $v(\bot_i) = v(\bot_j)$ for all $i, j = 1..3$, and $v(\bot_i) = 1$ for all $i = 4,5$. 
\end{example}


Formulas we write in the remainder are over signature $\sigma\cup Null$, where $\sigma$ is the database schema. In any incomplete database $D$ over $\sigma \cup Null$, $Null$ is always interpreted by the set of nulls occurring in $D$ (in accordance with the semantics of the {\sc{SQL}} construct \verb|IS NULL|). I.e. we allow rewritings to test whether a database element is null or not. 

For $\gamma(\bar y)$ a conjunction of equality atoms, using $equiv_\gamma$ we define a new formula $comp_\gamma(\bar{y})$ stating the existence of a consistent valuation that collapses all equivalent elements of a tuple:
$$
comp_{\gamma}(\bar{y})~:=~
\forall z z'(equiv_{\gamma}(\bar{y},z,z') \wedge \neg Null(z) \wedge \neg Null(z')\rightarrow z=z')
$$


\begin{proposition}\label{comp}
Let $\gamma(\bar y)$ be a conjunction of equality atoms, $D$ a database, 
and $\nu(\bar{y})=\bar{t}$ an assignment over $\adom(D) \cup \adom(\gamma)$, then $D \models comp_{\gamma} (\bar{t})$ 
if and only if there exists a consistent valuation $v$ of nulls such that  $v(D) \models \gamma(v(\bar{t}))$.
Moreover if  such valuation exists, there exists one further satisfying $v(s) = v(s')$ iff $ D \models equiv_\gamma(\bar t, s,s')$, for all $s, s' \in \adom(D)\cup\adom(\gamma)$ .
\end{proposition}

\begin{proof}
$\Rightarrow$ Assume $D \models comp_{\gamma} (\bar{t})$. Then $\forall c, c'$ 
constants, $D\models equiv_\gamma(\bar t, c, c')$ implies $c=c'$. As $\{(s,s') | D \models equiv_\gamma(\bar t, s,s')\}$ is an equivalence relation over $\adom(D) \cup \adom(\gamma)$,  its equivalence classes form a partition of this set. 
In each equivalence class there is at most one constant, so we define a valuation $v$ mapping all nulls of a class to the unique constant of that class (or to a new fresh constant if the class does not contain any).
Note that $v$ has the property that $v(s) = v(s')$ iff $ D \models equiv_\gamma(\bar t, s,s')$; this allows to prove that $v$ is a consistent valuation. 

In fact consider an arbitrary EGD $(\varphi(\bar{u})\wedge \psi(\bar u)) \rightarrow z=z')$ in $\Sigma$, and assume $v(D) \models \varphi(\mu(\bar{u}))\wedge \psi(\mu(\bar u))$ for some $\mu$; we prove $\mu(z) = \mu(z')$. 
Since $\varphi$ is a conjunction of atoms with no constants and no repeated variables, there exists $\mu'$ such that $D \models \varphi(\mu'(\bar u))$ and $v(\mu'(\bar u)) = \mu(\bar u)$.
Take any equality $w = w'$ in $\psi(\bar u)$, we show that $D \models equiv_\gamma(\bar t, \mu'(w), \mu'(w'))$. In fact because $v(D) \models \psi(\mu(\bar u))$ we have $\mu(w) = \mu(w')$, and therefore $v(\mu'(w)) = v(\mu'(w'))$. By definition of $v$ then $D \models equiv_\gamma(\bar t, \mu'(w),\mu'(w'))$. 
By the last rule of the Datalog program defining $equiv_\gamma$, we thus have that $D \models equiv_\gamma(\bar t, \mu'(z), \mu'(z'))$; then again by definition of $v$ we have $v(\mu'(z))=v(\mu'(z')$, and therefore $\mu(z) = \mu(z')$. This shows that $v(D)$ satisfies all the EGDs, thus $v$ is consistent.

We have proved that $v$ satisfies the characterisation of Lemma~\ref{intermediate}, and can conclude  that $v(D) \models \gamma(v(\bar t))$.

$\Leftarrow$ Assume $v$ is a consistent valuation and $v(D) \models \gamma(v(\bar{t}))$.
Let $s,s'$ such that $D \models equiv_\gamma(\bar{t},s,s')\wedge \neg Null(s') \wedge  \neg Null(s)$. By Lemma~\ref{intermediate} we have $v(s)=v(s')$.  Moreover $s,s'$ are both constants, then $s=s'$. Hence $D \models equiv_\gamma(\bar{t},s,s')\wedge \neg Null(s') \wedge  \neg Null(s)\rightarrow s=s'$, i.e., $D \models comp_{\gamma} (\bar{t})$.
\end{proof}

We are now ready to define a formula capturing the inclusion of supports between two conjunctions of equality atoms, which will be a crucial ingredient in our rewriting.
Let $\gamma(\bar{x})$ and $\gamma'({\bar y})$ be conjunctions of equality atoms with $\adom(\gamma) = \adom(\gamma')$. We  define :
$$ imply_{\gamma, \gamma'} (\bar{x},\bar{y})~:=~
\forall z z'~(equiv_{\gamma'}(\bar{y},z,z')\rightarrow equiv_{\gamma}(\bar{x},z,z'))
$$
\comment{ 
\begin{align}
& imply_{\gamma, \gamma'} (\bar{x},\bar{y})~:=~\nonumber\\
&\forall z z'~(equiv_{\gamma'}(\bar{y},z,z')\rightarrow equiv_{\gamma}(\bar{x},z,z'))\nonumber
\end{align}
}
Using Proposition~\ref{comp} 
 and Lemma~\ref{intermediate} we obtain :

\begin{proposition}\label{val}
Let $\gamma(\bar{x})$, $\gamma'({\bar y})$ be conjunctions of equality atoms with $\adom(\gamma) = \adom(\gamma')$, $D$ a database  
and $\nu(\bar{y})=\bar{t}$, $\nu'(\bar{y})=\bar{t'}$ assignments over $\adom(D) \cup \adom(\gamma)$.
Then $D \models imply_{\gamma, \gamma'}(\bar t, \bar t') \vee \neg comp_\gamma(\bar t)$ iff for all consistent valuations $v$, one has $v(D) \models \gamma(v(\bar t))$ implies $v(D) \models \gamma'(v(\bar t'))$.
\end{proposition}

\comment{
Before we show Proposition~\ref{val} we first show that, in order to test inclusion of supports of two equality formulas, one can restrict to single valuations collapsing just what is needed.

\begin{definition}[Tight valuation]
Let $\gamma(\bar y)$ be a conjunction of equality atoms, $D$ a database and $\nu(\bar{y})=\bar{t}$ an assignment over $\adom(D) \cup \adom(\gamma)$. A valuation $v$ of $D$ is called \emph{tight for} $\nu$ \emph{and} $\gamma$ if, for all $s, s' \in \adom(D)\cup\adom(\gamma)$, we have $v(s) = v(s')$ iff $D \models equiv_{\gamma}(\bar{t},s,s')$. 
\end{definition}
 
By Lemma~\ref{intermediate}, any tight valuation $v^*$ for $\nu$ and $\gamma$ satisfies $v^*(D) \models \gamma(v^*(\bar{t}))$. 
It is also easy to see that a tight valuation for $\nu$ and $\gamma$ exists whenever there is a valuation $v$ with $v(D) \models \gamma(v(\bar{t}))$. 
In fact if such a $v$ exists, by Proposition~\ref{comp}, $D \models comp_{\gamma} (\bar{t})$. Then  for each $s \in \adom(D) \cup\adom(\gamma)$ there is at most one constant $c$ such that $D \models equiv(\bar t, s, c)$. In addition we associate to each equivalence class $\mathcal C$ of the relation $\{(s,s') ~|~ D \models equiv_{\gamma}(\bar t, s, s')\}$, a new fresh constant $c_{\mathcal C}$ outside $\adom(D) \cup \adom(\gamma)$. Then a tight valuation $v^*$ for $\bar t$ and $\gamma$ can be defined as follows. For $s\in \adom (D)$, if $D \models equiv_{\gamma} (\bar t, s, c)$, for some constant $c$, then $v^*(s) = c$; otherwise $v^*(s) = c_{\mathcal{C}}$ where $\mathcal{C}$ is the equivalence class of $s$.

We can also characterise in terms of tight valuations the fact that for all valuations $v$, $v(D) \models \gamma(v(\bar t)) \Rightarrow v(D) \models \gamma'(v(\bar t'))$.

\begin{lemma}\label{tight}
Let $D$ be a database, $\gamma(\bar y)$, $\gamma'(\bar y)$ conjunctions of equality atoms with $\adom(\gamma) = \adom(\gamma')$, $\nu(\bar{y})=\bar{t}$, $\nu'(\bar{y})=\bar{t'}$ assignments over $\adom(D) \cup \adom(\gamma)$ and $v^*$ a tight valuation of $D$ w.r.t. $\nu$ and $\gamma$. Then $v^*(D) \models \gamma'(v^*(\bar t'))$ iff for all valuations $v$, $v(D) \models \gamma(v(\bar t))$ implies $v(D) \models \gamma'(v(\bar t'))$.
\end{lemma}

 \begin{proof}
$\Rightarrow$ Assume $v^*(D) \models \gamma'(v^*(\bar t'))$ and let $v$ be a valuation such that $v(D) \models \gamma(v(\bar t))$. 
We want to show $v(D) \models \gamma'(v(\bar t'))$. By Lemma~\ref{intermediate} it is enough to show that $\forall s, s' \in \bar{t'}$, $D \models equiv_{\gamma'}(\bar{t'},s,s')$ implies $v(D) \models v(s)=v(s')$. So let $s, s' \in \bar{t'}$ such that $D \models equiv_{\gamma'}(\bar{t'},s,s')$. As $v^*(D) \models \gamma'(v^*(\bar t'))$, by Lemma~\ref{intermediate}, $v^*(D) \models v^*(s)=v^*(s')$. Now $v^*$ is tight w.r.t. $\nu$ and $\gamma$, so $D \models equiv_{\gamma}(\bar{t},s,s')$. As $v(D) \models \gamma(v(\bar t))$, by Lemma~\ref{intermediate} it follows that $v(D) \models v(s)=v(s')$.

$\Leftarrow$ 
Assume for all valuations $v$, $v(D) \models \gamma(v(\bar t))$ implies $v(D) \models \gamma'(v(\bar t'))$. By Lemma~\ref{intermediate}, $v^*$ being tight for $\nu$ and $\gamma$, we have $v^*(D) \models \gamma(v^*(\bar{t}))$ and so by our assumption $v^*(D) \models \gamma'(v^*(\bar{t'}))$.
\end{proof}

We now have all the ingredients necessary to prove Proposition~\ref{val}.
} 

\begin{proof}
$\Rightarrow$ Assume $D \models imply_{\gamma, \gamma'}(\bar t, \bar t') \vee \neg comp_\gamma(\bar t)$. If $D \models \neg comp_\gamma(\bar t)$ then by Proposition~\ref{comp}, there is no consistent valuation $v$ such that $v(D) \models \gamma(v(\bar{t}))$ and so the implication trivially holds. Now assume $D \models imply_{\gamma, \gamma'}(\bar t, \bar t')$, i.e. : 
$$D \models \forall z z'~(equiv_{\gamma'}(\bar{t'},z,z')\rightarrow equiv_{\gamma}(\bar{t},z,z'))$$

We want to show that for all consistent valuations $v$ of nulls such that $v(D) \models \gamma(v(\bar t))$, one also has $v(D) \models \gamma'(v(\bar t'))$. Consider a valuation $v$ of nulls such that $v(D) \models \gamma(v(\bar t))$. Using Lemma~\ref{intermediate}, it is enough to show that $\forall s,s'$ such that $D \models equiv_{\gamma'}(\bar{t'},s,s')$ one has $v(s)=v(s')$. So assume $D \models equiv_{\gamma'}(\bar{t'},s,s')$, by our assumption it follows that $D \models equiv_{\gamma}(\bar{t},s,s'))$, and therefore, again by Lemma~\ref{intermediate},  $v(s)=v(s')$.

$\Leftarrow$ Assume for all consistent valuations $v$ of nulls such that $v(D) \models \gamma(v(\bar t))$, one also has $v(D) \models \gamma'(v(\bar t'))$. By Proposition~\ref{comp}, if there is no such valuation, then $D \models  \neg comp_\gamma(\bar t)$. So assume now there is one such valuation. This entails that in particular, there exists a consistent $v^*$ satisfying $v^*(D) \models \gamma(v^*(\bar t))$ and $v^*(s) = v^*(s')$ iff $ D \models equiv_\gamma(\bar t, s,s')$. By our assumption we have $v^*(D) \models \gamma'(v^*(\bar t'))$. Hence, by Lemma~\ref{intermediate}, $\forall s,s'$ such that $D \models equiv_{\gamma'}(\bar{t'},s,s')$ one has $v^*(s)=v^*(s')$. By the properties of $v^*$ mentioned above this implies $D \models  equiv_\gamma(\bar t,s,s')$. We have thus shown that $D \models imply_{\gamma, \gamma'}(\bar t, \bar t')$.
\end{proof}

So far, we have dealt with equality subqueries and we have characterized the emptiness and inclusion of their supports (cf. Proposition~\ref{comp} and Proposition~\ref{val}, respectively). We can now use this machinery to characterize the support of a BCCQ. 
We start by expressing membership in the support of an individual CQ~:

\begin{lemma}\label{SuppUCQ}
Let $D$ be a database, 
$v$ a consistent valuation of $D$ and $Q(\bar{x})$
a conjunctive query
in NRV-normal form, with relational subquery $q(\bar w)$ and equality subquery $\gamma(\bar x, \bar w)$.
Then $v \in \supp{(Q,D, \bar{r})}$ if and only 
there exists $\bar{s}$ such that  $D \models q(\bar{s})\wedge comp_{\gamma} (\bar{r} \bar{s})$ and $v(D) \models \gamma(v(\bar{r}\bar{s}))$. 
\end{lemma}

\begin{proof}
$\Rightarrow$ Assume  $v \in Supp(Q,D, \bar{r})$, i.e. $v$ is consistent and  $v(\bar{r}) \in Q(v(D))$ and so there exists 
an assignment $\mu$ of $\bar x \bar w$ over $\adom(v(D)) \cup \adom(Q)$ such that $\mu(\bar x) = v(\bar r)$ 
and $v(D) \models q(\mu(\bar w)) \wedge \gamma(\mu(\bar x\bar w))$.

%

Recall that $q(\bar{w})$ is a conjunction of relational atoms, with no constants and where each one of the free variables $\bar{w}$ has at most one occurrence in $q$. Thus there exists a mapping $\nu$ of $\bar w$ over $\adom(D)$ such that $D \models q(\nu(\bar w))$ and $v(\nu(\bar w))=\mu(\bar w)$. We let $\bar s = \nu(\bar w)$.
Recall that $\bar x$ and $\bar w$ do not share variables, so we can extend the mapping $\nu$ by setting $\nu(\bar x) = \bar r$. We thus have that $\nu(\bar x \bar w) = \bar r \bar s$ and 
$v(\bar r \bar s) = \mu(\bar x \bar w)$.
It follows that $v$ is a consistent valuation for which $v(D) \models \gamma(v(\bar r\bar s))$; then by Proposition~\ref{comp} we also have $D\models comp_\gamma(\bar r\bar s)$.

$\Leftarrow$ 
Since $D\models q(\bar s)$ we have $v(D) \models q(v(\bar s))$. Moreover by the hypothesis $v(D) \models \gamma (v(\bar r \bar s))$, thus $v(D) \models Q(v(\bar r))$.
\end{proof}

 In the remainder we consider BCCQs $Q(\bar{x})~:=~Q_1(\bar{x}) \vee \ldots \vee Q_n(\bar{x})$ in NRV disjunctive normal form (DNF) where for all $1 \leq i \leq n$ :
 $$Q_i~:=Q_{i_0}(\bar{x}) \wedge \neg Q_{i_1}(\bar{x}) \wedge \ldots \wedge \neg Q_{i_m}(\bar{x})$$
 \noindent and for all $1 \leq j \leq m$ :
$$Q_{i_j}~:=\exists \bar{w}_{i_j} q_{i_j}(\bar{w}_{i_j}) \wedge \gamma_{i_j}\text{ with }\gamma_{i_j} := e_{i_j}(\bar{x}\bar{w}_{i_j})$$
For convenience, we assume w.l.o.g every conjunction of literals to be of the same length $m$. 
We can also assume without loss of generality that for each $i$ we have $\adom(\gamma_{i_j}) = \adom(\gamma_{i_0})$ for all $j$. In fact we can always pad any $\gamma_{i_j}$ with dummy equalities $c = c$ to extend its active domain. 

Given a disjunct $Q_i$ in a BCCQ in DNF, we now define \emph{poss}$_{Q_i}$, encoding the set of possible answers to $Q_i$, and \emph{cons}$_{Q_i}$, checking the compatibility of an answer with the negative literals in $Q_i$. 
\begin{align}
poss_{Q_i}(\bar{x}\bar{w}) := q_{i_0}(\bar{w}) \wedge comp_{\gamma_{i_0}}(\bar{x}\bar{w}) \wedge cons_{Q_i}(\bar{x}\bar{w}) \nonumber
\end{align}
\begin{align}
cons_{Q_i}& \nonumber(\bar{x}\bar{w}) := \bigwedge_{\substack{1\leq j \leq m}} \forall \bar{w}'((q_{i_j}(\bar{w}') \wedge comp_{\gamma_{i_j}}(\bar{x}\bar{w}')) \rightarrow \neg imply_{\gamma_{i_0}, \gamma_{i_j}}(\bar{x}\bar{w}, \bar{x}\bar{w}')) \nonumber
\end{align}
Using these new formulae, we show that the non-emptiness of $\supp{(Q(\bar{x}), D, \bar{r})}$ can be expressed as the existence of a possible answer.


\begin{proposition}\label{suppBCCQ}
Let $D$ be a database and $Q(\bar{x})$ a DNF BCCQ in NRV normal form,
then $\supp{(Q(\bar{x}), D, \bar{r})} \neq \emptyset$ if and only if $D \models \bigvee_{\substack{1\leq i \leq n}} \exists \bar{w}$ $poss_{Q_i}(\bar{r}\bar{w})$.
\end{proposition}

\begin{proof}
$\Leftarrow$ Let $D \models  \bigvee_{\substack{1\leq i \leq n}}\exists \bar{w}$ $poss_{Q_i}(\bar{r}\bar{w})$, then there exists $1 \leq i \leq n$ and an assignment $\nu$ with $\nu(\bar{w})=\bar{s}$, $D \models q_{i_0}(\bar{s})\wedge comp_{\gamma_{i_0}}(\bar{r}\bar{s})$ and for all $1 \leq j\leq m$, for all $\bar{s'}$ such that $D\models q_{i_j}(\bar{s'})\wedge comp_{\gamma_{i_j}}(\bar{r}\bar{s'})$, one has $D \models \neg imply_{\gamma_{i_0},\gamma_{i_j}}(\bar r \bar s , \bar r \bar s')$. Since $D \models comp_{\gamma_{i_0}}(\bar{r}\bar{s})$, 
by Proposition \ref{comp} there exists a consistent valuation $v^*$, such that $v^*(D) \models \gamma_{i_0}(v^*(\bar{r}\bar{s}))$ and for all $s, s' \in \adom(D)\cup\adom(\gamma_{i_0})$, we have $v^*(s) = v^*(s')$ iff $D \models equiv_{\gamma_{i_0}}(\bar r \bar s,s,s')$.

Moreover we can prove the following claim  :

\begin{claim}
For all conjunction of equalities $\gamma'(\bar y)$ with $\adom(\gamma') = \adom (\gamma_{i_0})$ and all $\bar t$ over $\adom(D) \cup \adom(\gamma_{i_0})$, one has  $v^*(D) \models \gamma'(v^*(\bar t))$ iff for all consistent valuations $v$, $v(D) \models \gamma_{i_0}(v(\bar{r}\bar{s}))$ implies $v(D) \models \gamma'(v(\bar t))$.
\end{claim}
\begin{proof}[Proof of the Claim]
$\Rightarrow$ Assume $v^*(D) \models \gamma'(v^*(\bar t))$ and let $v$ be a consistent valuation such that $v(D) \models \gamma_{i_0}(v(\bar r \bar s))$. 
We want to show $v(D) \models \gamma'(v(\bar t))$. By Lemma~\ref{intermediate} it is enough to show that $\forall s, s' \in \adom(D) \cup\adom(\gamma_{i_0})$, $D \models equiv_{\gamma'}(\bar{t},s,s')$ implies $v(s)=v(s')$. So let $s, s'$ be such that $D \models equiv_{\gamma'}(\bar{t},s,s')$. As $v^*(D) \models \gamma'(v^*(\bar t))$, by Lemma~\ref{intermediate}, $v^*(s)=v^*(s')$. By the properties of $v^*$, so $D \models equiv_{\gamma_{i_0}}(\bar r \bar s,s,s')$. As $v(D) \models \gamma_{i_0}(v(\bar r \bar s))$, by Lemma~\ref{intermediate} it follows that $v(s)=v(s')$.

$\Leftarrow$ Assume for all consistent valuations $v$, $v(D) \models \gamma_{i_0}(v(\bar r \bar s))$ implies $v(D) \models \gamma'(v(\bar t))$. By Lemma~\ref{intermediate}, $v^*$ being consistent $v^*(D) \models \gamma'(v^*(\bar{t}))$.
\end{proof}

Now fix some arbitrary $j \geq 1$ and $\bar{s'}$ with $D\models q_{i_j}(\bar{s'})\wedge comp_{\gamma_{i_j}}(\bar{r}\bar{s'})$. By Proposition~\ref{val}, it follows from $D \models \neg imply_{\gamma_{i_0},\gamma_{i_j}}(\bar{r}\bar{s},\bar{r}\bar{s'}) \wedge comp_{\gamma_{i_0}}(\bar r \bar s)$ that there exists a consistent valuation $v'$ with $v'(D) \models \gamma_{i_0}(v'(\bar{r}\bar{s}))$ but $v'(D) \not \models \gamma_{i_j}(v'(\bar{r}\bar{s'}))$. By the above claim 
$v^*(D) \not \models \gamma_{i_j}(v^*(\bar{r}\bar{s'}))$.
In summary we have :

\begin{enumerate}
\item  [(i)] $D\models q_{i_0}(\bar{s})\wedge comp_{\gamma_{i_0}}(\bar{r}\bar{s})$ and $v^*(D)\models \gamma_{i_0}(v^*(\bar{r}\bar{s}))$ and so by Lemma \ref{SuppUCQ}, we have $v^* \in \supp{(Q_{i_0}(\bar{x}),D,\bar{r})}$, i.e., $v^*(D) \models Q_{i_0}(v^*(\bar{r}))$.
\item  [(ii)] For all $1 \leq j \leq m$ and assignment $\nu'$ with $\nu'(\bar{w})=\bar{s'}$, if $D \models q_{i_j}(\bar{s'}) \wedge comp_{\gamma_{i_j}}(\bar{r}\bar{s'})$ then $v^*(D) \not \models  \gamma_{i_j}(v^*(\bar{r}\bar{s'}))$ and so by Lemma \ref{SuppUCQ}, we have $v^* \not\in \supp{(Q_{i_j}(\bar{x}),D,\bar{r})}$, i.e., for all $1 \leq j \leq m$, $v^*(D) \models \neg Q_j(v^*(\bar{r}))$.
\end{enumerate}
This means we have $v^* \in \supp{(Q_{i_0}(\bar{x}) \wedge \neg Q_{i_1}(\bar{x})\wedge \ldots \wedge \neg Q_{i_m}(\bar{x}),D,\bar{r})}$ for all $1 \leq i \leq n$ and so $v^* \in \supp{(Q(\bar{x}),D,\bar{r})}$.

$\Rightarrow$ Let $v \in \supp{(Q(\bar{x}),D,\bar{r})}$, so $v$ is consistent and there is some $1 \leq i \leq n$ with:  (i) $v \in \supp{(Q_{i_0}, D, \bar{r})}$, (ii) for all $1 \leq j \leq m$, $v \not\in \supp{(Q_{i_j}, D, \bar{r})}$.
Using Lemma \ref{SuppUCQ} (i) implies that there exists $\bar{s}$ such that $D \models q_{i_0}(\bar{s})\wedge comp_{\gamma_{i_0}}(\bar{r}\bar{s})$ and $v(D) \models \gamma_{i_0}(v(\bar{r}\bar{s}))$. Again by Lemma \ref{SuppUCQ}, (ii) implies that for all $1 \leq j \leq m$ and $\bar{s'}$, if $D \models q_{i_j}(\bar{s'}) \wedge comp_{\gamma_{i_j}}(\bar{r}\bar{s'})$ then $v(D) \not \models  \gamma_{i_j}(v(\bar{r}\bar{s'}))$. This entails by Proposition \ref{val} that $D \models comp_{\gamma_{i_0}}(\bar{r}\bar{s}) \wedge \neg imply_{\gamma_{i_0},\gamma_{i_j}}(\bar{r}\bar{s},\bar{r}\bar{s'})$. This shows $D \models \bigvee_{\substack{1\leq i \leq n}} \exists \bar{w}$ $poss_{Q_i}(\bar{r}\bar{w})$.
\end{proof}

Now that we have defined the formula expressing for a BCCQ $Q$ non-emptiness of $\supp{(Q(\bar{x}), D, \bar{r})}$ (Proposition \ref{suppBCCQ}), we can easily define a rewriting for the problem \certainanswer$(Q)$. 
To do so, we rely on the fact that $\bar r \in \cert_\Sigma(Q,D)$ 
iff $\supp{(\neg Q, D, \bar r)} = \emptyset$. 

\begin{theorem}[Datalog rewriting]\label{certain}
Let D be a database whose schema contains a set of equality generating dependencies $\Sigma$, and let $Q(\bar{x})$ be a BCCQ in NRV-normal form. Let $Q' = Q_1'(\bar{x}) \vee \ldots \vee Q_n'(\bar{x})$ be $\neg Q$ in DNF normal form. Then $\bar r \in \cert_\Sigma(Q,D)$ if and only if 
$D \models \rho(\bar r)$ where $\rho(\bar x) = \bigwedge_{\substack{1\leq i \leq n}} \forall \bar{w}$ $\neg poss_{Q'_i}(\bar{x}\bar{w})$.
\end{theorem}

\begin{proof}
One has that $\bar r \in \cert_\Sigma(Q,D)$ iff $Supp(Q', D, \bar r) = \emptyset$. $Q'$ being still a BCCQ, Proposition~\ref{suppBCCQ} tells us that $Supp(Q', D, \bar r) =\emptyset$ iff  $D \models \bigwedge_{\substack{1\leq i \leq n}} \forall \bar{w}$ $\neg poss_{Q'_i}(\bar{r}\bar{w})$.
%
\end{proof}

\begin{corollary}
For each fixed BCCQ query $Q$ and a set of EGDs $\Sigma$, the complexity of \certainanswer$_(Q)$ is in PTIME.
\end{corollary}

\comment{
\begin{example}
Assume $\Sigma=\emptyset$. Take $Q(x) = R_1(x) \wedge \neg R_2(x)$ over $D = \{R_1(1), R_2(\bot_1)\}$. Is 1 an answer to $Q(x)$ ?\\
First we need to put $Q(x)$ in NRV Normal Form : \\$Q^{NRV} = \exists z_1 R_1(z_1) \wedge z_1 = x \wedge \neg (\exists z_2 R_2(z_2) \wedge z_2 = x)$.\\ Let $\nu$ be an assignment such that $\nu(z_1)=1$ and $\nu(z_2)=\bot_1$ and $v$ a valuation. By Proposition \ref{suppBCCQ}, 1 is an answer if $D \models R_1(1) \wedge comp_e(1, \{\}, 1) \wedge (R_2(\bot_1) \wedge comp_{e'}(1, \{\}, \bot_1) \rightarrow \neg imply_{e, e'}(1\{\}1, 1\{\}\bot_1))$. \\If the valuation $v$ sends $\bot_1$ to 1, $imply_{e, e'}(1\{\}1, 1\{\}\bot_1)$ is true (as $z_1 \equiv_e^\nu z_1$  and $z_1 \equiv_{e'}^\nu z_2$), thus rendering $R_2(\bot_1) \wedge comp_{e'}(1, \{\}, \bot_1) \rightarrow \neg imply_{e, e'}(1\{\}1, 1\{\}\bot_1))$ false and proving that $v \notin Supp(Q(x), D, 1)$. \\However, if $v(\bot_1) \neq 1$, then $R_2(\bot_1) \wedge comp_{e'}(1, \{\}, \bot_1) \rightarrow \neg imply_{e, e'}(1\{\}1, 1\{\}\bot_1))$ is true and we have that $v \in Supp(Q(x), D, 1)$. 
\\1 is the best but not a certain answer for Q(x) over D. 
\end{example}
}

%
%
\label{dat}

\subsection{Non-rewritability in FO}
%
%

The basic starting points for our investigation was the fact that $\cert_\Sigma(Q,D)=Q(\chase_\Sigma(D))$ for a CQ $Q$ and a set $\Sigma$ of FDs, for every database $D$. This remained true for unions of CQs, but failed for BCCQs, forcing us to produce a Datalog rewriting to obtain certain answers. But can a first-order rewriting be obtained instead? This would make it possible to produce certain answers using the core of SQL as opposed to its recursive features which do not always perform as well in practice. 

In this section we show that the answer, in general, is negative even for CQs (and thus for BCCQs). In the next section however we show that such rewritings can be obtained in FO for BCCQs whenever $\Sigma$ is empty. 

The main result of this section is the following.

\begin{theorem}
\label{no-fo-thm}
There exists a Boolean CQ $Q$ and single FD $\Sigma$ 
over a relational schema of binary and unary relations such that $\cert_\Sigma(Q,D)$ is not expressible as an FO query.
\end{theorem}

\noindent
\begin{proof}
Consider a schema with one binary relation $E$ and two unary relations $A$ and $B$. The only FD in $\Sigma$ is $\forall x \forall y \forall z\ \big(E(x,y) \wedge E(x,z) \to y =z\big)$; in other words, the first attribute of $E$ is a key. The query $Q$ is a Boolean CQ $\exists x\ (A(x)
\wedge B(x))$. 

To prove inexpressibility of $\cert_\Sigma(Q,\cdot)$ in FO, for each $n > 0$ we create two databases $D_n$ and $D_n'$. In both of them, $E$ is interpreted as a disjoint union $T_1\cup T_2$ where $T_1$ and $T_2$ are balanced binary trees of depth $n$ in which all nodes are distinct nulls. In both $A$ and $B$ are singleton sets. In $D_n$, the set $A$ contains a leaf of $T_1$ and $B$ contains a leaf of $T_2$. In $D_n'$, both $A$ and $B$ contain leaves of $T_1$ such that their only common ancestor in the tree is the root (in other words, they are leaves of subtrees rooted at different children of the root of $T_1$). 

Because of the constraint $\Sigma$, for every valuation $v$ such that the resulting database satisfies it we have that both $v(T_1)$ and $v(T_2)$ are chains. Indeed, consider any node $\bot$ with children $\bot_1,\bot_2$ in $T_i$. If $v(\bot_1)\neq v(\bot_2)$ then the resulting tuples $(v(\bot), v(\bot_1))$ and $(v(\bot),v(\bot_2))$ violate the constraint. Thus $v(\bot_1)=v(\bot_2)$ and applying this construction inductively we see that $v(T_i)$ is a chain. Hence, it has a single leaf, and thus $\cert_\Sigma(Q,D_n')$ is true, since $A$ and $B$ must be interpreted as that leaf. On the other hand, $\cert_\Sigma(Q,D_n)$ is false, since there is a valuation $v$ that sends $T_1$ and $T_2$ into two disjoint chains, and thus $A$ and $B$ are interpreted as two distinct elements. 

Assume now that $\cert_\Sigma(Q,\cdot)$ is rewritable as an FO sentence $\phi$. Then, for every $n > 0$, we have $D_n'\models\phi$ and $D_n\models\neg\phi$. We next show that such a sentence cannot exist, thereby proving non-FO-rewritability. 

Recall that in a database (with one binary relation, like considered here) a radius $r$ neighborhood of an element $a$ is its restriction to the set of all elements reachable from $a$ by a path of length at most $r$, where the path does not take into account the orientation of edges of $E$ (for example, if we have $E(a,b)$ and $E(c,a)$ then both $b$ and $c$ are in the radius $1$ neighborhood of $a$). When two neighborhoods, of elements $a$ and $b$, are isomorphic, it means that there is an isomorphism between them that sends $a$ to $b$. In other words, centers of neighborhoods are viewed as distinguished elements when it comes to defining neighborhoods. It is known that each first order sentence $\psi$ is Hanf-local \mycite{fsv95}: that is, there exists a number $r > 0$ such that for any two databases $D_1$ and $D_2$, if there is a bijection $f$ between $D_1$ and $D_2$ such that the radius $r$ neighborhoods of $a$ in $D_1$ and $f(a)$ in $D_2$ are isomorphic then $D_1$ and $D_2$ agree on $\psi$, i.e. either both satisfy it or both do not.

Now let $r$ be such a number for the sentence $\phi$ we assumed exists. Consider $D_n$ and $D_n'$ and let $T_{1a},T_{1*}$ be the subtrees of the root of $T_1$ in $D_n$ such that the first contains $A$ while the second contains neither $A$ not $B$, and let $T_{2b},T_{2*}$ be defined similarly for subtrees of the root of $T_2$ with respect to $B$. In $D_n'$ we define $T_{1a}',T_{1b}'$ as subtrees of the root of the tree containing $A,B$ such that the first contains the $A$ leaf and the second contains the $B$ leaf, while $T_{2*}',T_{2**}'$ be the subtrees of the root of the tree having neither $A$ nor $B$ elements. Then it is easy to see that the following pairs of trees are isomorphic: $T_{1a}$ and $T_{1a}'$, $T_{2b}$ and $T_{1b}'$, $T_{1*}$ and $T_{2*}'$, $T_{2*}$ and $T_{2**}'$.

We now define the bijection $f$ as the union of those isomorphisms plus mapping roots of trees $T_i$ in $D$ into roots of $T_i$ in $D'$. It is an immediate observation that if $n > r+1$ (i.e., leaves are not in the radius $r$ neighborhood of children of roots) then $f$ satisfies the condition that neighborhoods of $a$ and $f(a)$ of radius $r$ are isomorphic. This would tell us that $D_n$ and $D_n'$ agree on $\phi$ but we know they do not. This contradiction completes the proof. 
\end{proof}

As a corollary to the proof, we obtain the following result showing that non-recursive SQL is incapable of computing $\cert_\Sigma(Q,D)$ in the setting of Theorem \ref{no-fo-thm}.

\begin{corollary}
\label{no-aggr-cor}
There exists a Boolean CQ $Q$ and single FD $\Sigma$ 
over a relational schema of binary and unary relations such that $\cert_\Sigma(Q,D)$ is not expressible in the basic {\tt SELECT-FROM-WHERE-GROUP BY-HAVING} fragment of SQL with arbitrary aggregate functions.
\end{corollary}

This is due to the fact that queries in this fragment of SQL with grouping and aggregation can be translated into a logic with aggregate functions \mycite{tcs03} which itself is known to be Hanf-local \mycite{HLNW01}. 
%
%
\label{nofo}

\subsection{FO rewriting for Certain Answers}
%
%

We now focus on the special case where $\Sigma$ is empty. First notice that the only Datalog component in our rewriting was the $equiv_{\gamma}$ formula; moreover notice that  without constraints, $equiv_{\gamma}$ simply computes a reflexive symmetric transitive closure. More precisely, for a given $\bar t = \nu(\bar y)$, one has that  $equiv_{\gamma}(\bar t, s, s')$ holds in $D$ iff $(s,s')$ belongs to the reflexive symmetric transitive closure of $\{(\nu(x), \nu(w)) ~|~ x = w \in \gamma(\bar y)\}$ over $\adom(D) \cup\adom(\gamma)$.

\begin{definition}
Given a database $D$, a conjunction of equality atoms $\gamma(\bar{y})$ and an assignment $\nu: \bar y \cup \adom(\gamma) \rightarrow\adom(D)\cup \adom(\gamma)$ preserving constants,
we say that $u,u' \in \adom(D)\cup\adom(\gamma)$  are equivalent w.r.t. $\gamma$ and $\nu$ and write $u \equiv_\gamma^\nu u'$, 
if either $u=u'$ or $(u, u')$ belongs to the reflexive symmetric transitive closure of $\{(\nu(x), \nu(w)) ~|~ x = w \in \gamma\}$.
\end{definition}

As $\Sigma$ is empty, we can rewrite as follows the $equiv_{\gamma}$ formula in FO, 
where $m$ is the number of equivalence classes of $\sim_{\gamma}$ : 
\begin{align}
&\equivfo_{\gamma}(\bar{y},z,z')~:=~ z=z'~\vee~ \nonumber\\
&\bigvee_{\substack{u_1, v_1\ldots u_m, v_m \in~\bar{y}~\cup~\adom(\gamma)~|~\nonumber\\
u_i \sim_\gamma v_i \textrm{ for all } 1 \leq i \leq m}}\!\!\!\!\!\!\!\!\!\!\!(z=u_1 \wedge z'=v_m \wedge \bigwedge_{1\leq i < m}v_i=u_{i+1})
\end{align}

\begin{proposition}\label{equiv}
Given an incomplete database $D$, a conjunction of equality atoms $\gamma(\bar{y})$ 
and an assignment $\nu(\bar{y})=\bar{t}$ over $\adom(D) \cup \adom(\gamma)$, 
given $s,s'$ in $\bar t\cup\adom(\gamma)$, we have that
$D \models \equivfo_{\gamma} (\bar{t},s,s')$
if and only if $s\equiv_\gamma^\nu s'$.
\end{proposition}

Intuitively this holds because each disjunct of $\equivfo_\gamma(\bar t, s, s') $ corresponds to a possible derivation of $(s, s')$ in the reflexive symmetric transitive closure of $\{(\nu(x), \nu(w)) ~|~ x = w \in \gamma\}$, and one can prove that there is a bound only depending on $\gamma$ on the number of steps of this derivation. 

\begin{proof}
To start with we naturally extend $\nu$ to be the identity on $\adom(\gamma)$.
Assume first $D \models \equivfo_{\gamma} (\bar{t},s,s')$. If $s=s'$ then $s\equiv_\gamma^\nu s'$. Now assume $s \neq s'$. Then there exist variables and/or constants $u_1, v_1\ldots u_m, v_m \in~\bar{y}~\cup~\adom(\gamma)$ with $u_i \sim_\gamma v_i$ for all $i$, such that $s = \nu(u_1), s' = \nu(v_m)$ and $\nu(v_i) = \nu (u_{i+1})$ for all $i < m$.
Clearly $u_i\sim_\gamma v_i$ implies $\nu(u_i) \equiv_\gamma^\nu \nu(v_i)$.
Then  $\nu(u_i) \equiv_\gamma^\nu \nu (u_{i+1})$ for all $i < m$. We conclude by transitivity that $s= \nu(u_i) \equiv_\gamma^\nu \nu(u_m) \equiv_\gamma^\nu \nu(v_m)=s'$, and therefore  $s \equiv_\gamma^\nu s'$.

Assume now that $s \equiv_\gamma^\nu s'$. If $s = s'$ then clearly $D \models \equivfo_{\gamma} (\bar{t},s,s')$. Thus assume $s \neq s'$. We proceed by induction on the number of transitive closure steps needed to derive $(s,s')$ starting for the base relation $\{(\nu(x), \nu(w)) | x = w \in \gamma\}$.
In the base case 
$(s,s') = (\nu(x), \nu(w))$ for some equality $x=w \in \gamma$. 
Then $D$ satisfies the following disjunct of $\equivfo_{\gamma} (\bar{t},s,s')$ : take $u_1 = x$, $v_1 = w$, $u_i = v_i = w$ for all $i = 2..m$ (this is a disjunct since $u_i \sim_\gamma v_i$ for all $i = 1..m$). The disjunct is satisfied since $s = \nu(x) = \nu(u_1)$, $s' = \nu(w) = \nu(v_m)$, and for all $i = 1..m-1, \nu(v_i) = \nu(u_{i+1}) = \nu(w)$.

In the general inductive case, there exists $r$ such that $(r, s') = (\nu(x), \nu(w))$ for some equality $x = w$ (or $w=x$) $\in \gamma$, with $s \equiv_\gamma^\nu r$ derived at the previous step. 
By the induction hypothesis $D \models \equivfo_\gamma(\bar t, s, r)$. We can assume $s \neq r$ since otherwise $(s,s')$ would be  in the base relation.
Therefore $D$ satisfies one of the disjuncts of $\equivfo_\gamma(\bar t, s, r)$. Then there exists a sequence of $m+1$ pairs  in $\bar y \cup \adom(\gamma)$ 

$$ (u_1, v_1) (u_2, v_2) \dots (u_m, v_m) (u_{m+1}, v_{m+1})$$

such that 
\begin{itemize}
\item $u_{m+1} = x$ and $v_{m+1} = w$,
\item $u_i \sim_\gamma v_i, i =1..m+1$,
\item  $s = \nu(u_1)$, $r = \nu(v_m)$, $s' = \nu(v_{m+1})$,   
\item $\nu(v_i) = \nu(u_{i+1})$, for all $i \leq m$, 
\end{itemize}

We now show that from this sequence of pairs one can construct another one of exactly $m$ pairs, $(u'_i, v'_i), i =1..m$ still connecting $s$ ans $s'$, i.e. such that  :
\begin{itemize}
\item [(a)] $u'_i \sim_\gamma v'_i, i =1..m$
\item [(b)] $s = \nu(u'_1)$, $s' = \nu(v'_{m})$ 
\item [(c)] $\nu(v'_i) = \nu(u'_{i+1})$, for all $i < m$.
\end{itemize}

The idea is to first cut the sequence $(u_i, v_i), i =1..m+1$, removing at least one pair, then pad it to size $m$ if necessary.

In order to cut the original sequence, remark that it contains $m+1$ pairs where $m$ is the number of $\sim_\gamma$ equivalence classes. Thus there exist $i < j$ such that $u_i \sim_\gamma u_j$. We remove from the sequence all elements between $u_i$ and $v_j$ (excluded), the new sequence is 
$$(u_1, v_1) ... (u_{i-1}, v_{i-1}) (u_i, v_j) (u_{j+1}, v_{j+1}) ...(u_{m+1}, v_{m+1})$$ Note that this sequence satisfies (a) (b) and (c) above since $u_i \sim_\gamma u_j \sim_\gamma v_j$. Let the new sequence contain $k$ pairs. We know $k \leq m$ because we have removed at least one pair from the original sequence (recall $i < j$).   If $k < m$ we pad the sequence on the right with $m-k$ pairs $(v_{m+1}, v_{m+1})$. The new sequence still satisfies (a), (b) and (c), therefore the corresponding disjunct of $\equivfo_\gamma(\bar t, s, s')$ is satisfied by $D$.
\end{proof}

\begin{example}\label{ex::equiv}
Let $\gamma:=y_1=y_2 \wedge z=x$ be the equality subquery of the query $Q(x)$ in Example \ref{ex::nrv}. Up to logical equivalence, $\equivfo_{\gamma}(y_1,y_2,z,x,w,w')$ contains precisely the disjuncts $w=w'$, $w=y_1 \wedge w'=y_2$, $w=z \wedge w'=x$, $w=y_1 \wedge w'=x \wedge y_2=z$, plus all disjuncts obtained from them by applying one or more of the following transformations : switch $w$ and $w'$, switch $y_1$ and $y_2$, switch $x$ and $z$. Let $D$ be the database from Example \ref{ex}, then we have for instance $D \models \equivfo_{\gamma}(1,\bot_2,\bot_2,1,a,a')$ and $D \models \equivfo_{\gamma}(1,\bot_2,\bot_2,\bot_2,a,a')$ for all $a,a' \in \{1,\bot_2\}$. Similarly $D \models \equivfo_{\gamma}(\bot_1,\bot_2,\bot_2,1,a,a')$ for all $a,a' \in \{1,\bot_1,\bot_2\}$.
\end{example}

As a consequence of 
Proposition~\ref{equiv}, for fixed $\gamma$ and $\bar t$, the relation $\{(s,s') ~|~ D \models \equivfo_{\gamma}(\bar t, s, s')\}$ is an equivalence relation over $\adom(D) \cup \adom(\gamma)$ where each element of $\adom(D)$ neither in $\bar t$ nor in $\adom(\gamma)$ forms a singleton equivalence class. 

As in section \ref{dat}, $ \equivfo_{\gamma}$ selects precisely the pairs of elements of a tuple that a valuation needs to collapse to satisfy a set of equalities.

As a consequence, we can rewrite in FO the formula poss$_{Q_i}$ of subsection \ref{dat} encoding the set of possible answers to $Q_i$. It is enough to replace each occurrence of the Datalog $equiv_{\gamma}(\bar{y},z,z')$ program in it by $equiv^{FO}_{\gamma}(\bar{y},z,z')$. We denote by poss$^{FO}_{Q_i}$ the rewriting so obtained.

\begin{theorem}[FO rewriting]\label{certainFO}
Let D be a database, $\Sigma=\emptyset$ and let $Q(\bar{x})$ be a BCCQ in NRV-normal form. Let $Q' = Q_1'(\bar{x}) \vee \ldots \vee Q_n'(\bar{x})$ be $\neg Q$ in DNF normal form. Then $\bar r \in \cert_\Sigma(Q,D)$ if and only if 
$D \models \rho(\bar r)$ where $\rho(\bar x) = \bigwedge_{\substack{1\leq i \leq n}} \forall \bar{w}$ $\neg poss^{FO}_{Q'_i}(\bar{x}\bar{w})$.
\end{theorem}
Note that tractability of BCCQ was already proved in \mycite{gheerbrant-libkin:xml} using tableau based methods. We now refine complexity as follows.
\begin{corollary}
For each fixed BCCQ query $Q$, the complexity of \certainanswer$_(Q)$ is in DLOGSPACE whenever ~$\Sigma=\emptyset$.
\end{corollary}

%
%
\label{fo}

\subsection{FO Rewriting for Best Answers }
%
%
Considering arbitrary \FO-queries brought us an intrinsic intractability result for all variants of best answers. This motivates restricting to unions of conjunctive queries, for which a polynomial time evaluation algorithm (in data complexity) already exists \mycite{libkin:zero}. The resolution based procedure is however in sharp contrast with 
 na\"ive evaluation, which allows to compute certain answers to unions of conjunctive queries via usual model checking. We thus initiate a descriptive complexity analysis of the best answers problem, showing that for unions of conjunctive queries, it can essentially be reduced - modulo a preprocessing of the query - to (na\"ive) evaluation of an $\FO$-formula.

Given a union of conjunctive queries $Q$, our starting point towards an $\FO$-rewriting for best answers is finding an $\FO$-formula $Q_\subseteq(\bar{x}, \bar{y})$ encoding the inclusion of supports, i.e. selecting tuples $\bar s, \bar t$ over $\adom(D)\cup\adom(Q)$ iff $\supp(Q,D,\bar s) \subseteq \supp(Q,D,\bar t)$.
From $Q_\subseteq$ one can easily define an $\FO$-formula selecting precisely all best answers to $Q$ on $D$:
\begin{align}\label{bestq}
best_Q (\bar x) := \forall \bar{y}(Q_{\subseteq}(\bar{x}, \bar{y}) \rightarrow Q_{\subseteq}(\bar{y}, \bar{x}))
\end{align}

As in section \ref{dat}, we assume all CQs to be in NRV normal form.  We can thus first concentrate on the support of equality subqueries. This will be encoded in \FO~ and then integrated in the rewriting of the whole conjunctive query.

We now go back to an arbitrary union of conjunctive queries of vocabulary $\sigma$ in NRV-normal form :
\begin{align}
Q(\bar{x}):=\bigvee\limits_{1 \leq i \leq n}Q_i(\bar{x})\nonumber
\end{align}
 where each $Q_i$ is in NRV-normal form with relational subquery $q_i(\bar y_i, \bar z_i)$ and equality subquery $eq_i(\bar x, \bar y_i, \bar z_i)$.


Recall the formula $comp_{\gamma}$, defined in section \ref{dat}, stating the existence of a valuation that collapses all equivalent elements of a tuple : 

$$
comp_{\gamma}(\bar{y})~:=~
\forall z z'(\equivfo_{\gamma}(\bar{y},z,z') \wedge \neg Null(z) \wedge \neg Null(z')\rightarrow z=z')
$$

Notice that if $D \models comp_\gamma(\bar t)$ then for each $s \in adom(D)\cup\adom(\gamma)$ there exists at most one constant $c$ such that $D \models \equivfo_\gamma(\bar t, s, c)$. In fact if for constants $c_1$ and $c_2$, 
$D \models \equivfo_\gamma(\bar t, s, c_1)$ and $D \models \equivfo_\gamma(\bar t, s, c_2)$,  by transitivity $D \models \equivfo_\gamma(\bar t, c_1, c_2)$, implying $c_1 = c_2$.

\begin{example}
Let $D$ and $\gamma$ be as in Example \ref{ex::equiv}. Consider $comp_\gamma(y_1,y_2,z,x)$. Given the tuples selected by $\equivfo_\gamma$ in Example \ref{ex::equiv}, we can conclude that $D \models comp_\gamma(1,\bot_2,\bot_2,1)$. 
\end{example}


We now define a formula capturing the inclusion of supports between two conjunctions of equality atoms, which will be a crucial ingredient in our rewriting.

Let $\gamma(\bar{x})$ and $\gamma'({\bar y})$ be conjunctions of equality atoms with $\adom(\gamma) = \adom(\gamma')$. We  define :

$$
imply_{\gamma, \gamma'} (\bar{x},\bar{y})~:=~\forall z z'~(\equivfo_{\gamma'}(\bar{y},z,z')\rightarrow \equivfo_{\gamma}(\bar{x},z,z'))
$$

\begin{example}\label{ex::imply}
Let $\gamma$ and $D$ be as in Example \ref{ex::equiv}. Let $\gamma':=y'_1=y'_2 \wedge z'=x'$, then it follows from Example \ref{ex::equiv} that $D \models imply_{\gamma\gamma'}(\bot_1\bot_2\bot_21,1\bot_2\bot_2\bot_2)$ and $D \models imply_{\gamma\gamma'}(1\bot_2\bot_21,1\bot_2\bot_2\bot_2)$.
\end{example}



By combining Propositions~\ref{val}~and~\ref{comp} we get :

\begin{corollary}\label{imply}
Let $\gamma(\bar y)$, $\gamma'(\bar y)$ be conjunctions of equality atoms with $\adom(\gamma) = \adom(\gamma')$, $D$ a database and $\nu(\bar{y})=\bar{t}$, $\nu'(\bar{y})=\bar{t'}$ assignments over $\adom(D) \cup \adom(\gamma)$. If $D \models comp_{\gamma}(\bar{t}) \wedge imply_{\gamma, \gamma'}(\bar{t},\bar{t'})$, then $D \models comp_{\gamma'}(\bar{t'})$.
\end{corollary}

\begin{proof}
Assume $D \models comp_{\gamma}(\bar{t}) \wedge imply_{\gamma, \gamma'}(\bar{t},\bar{t'})$, i.e., $D \models \forall z z'~(\equivfo_{\gamma}(\bar{t},z,z') \wedge \neg Null(z) \wedge \neg Null(z') \rightarrow z=z')$
and $D \models \forall z z'~(\equivfo_{\gamma'}(\bar{t'},z,z')\rightarrow \equivfo_{\gamma}(\bar{t},z,z'))$. 
Now let $s,s' \in\adom(D) \cup \adom(\gamma)$ with $D \models \equivfo_{\gamma'}(\bar{t'},s,s') \wedge \neg Null(s) \wedge \neg Null(s')$. As $D \models imply_{\gamma, \gamma'}(\bar{t},\bar{t'})$ it follows that $D \models \equivfo_{\gamma}(\bar{t},s,s')$ and so $D \models  \neg Null(s) \wedge \neg Null(s') \rightarrow s=s'$ now follows from $D \models comp_{\gamma}(\bar{t})$. Hence $D \models comp_{\gamma'}(\bar{t'})$.
\end{proof}

We are now ready to define the $\FO$-formula encoding the inclusion of supports.
\begin{align}
Q_\subseteq(\bar{x}, \bar{x'}) :=  \bigwedge\limits_{1\leq i \leq n}  ~~(&\forall \bar{y}\bar{z} ((q_i(\bar{y},\bar{z})\wedge comp_{eq_i}(\bar{x},\bar{y},\bar{z})) \rightarrow\nonumber\\
&\bigvee\limits_{1 \leq j \leq n}\exists \bar{y}' \bar{z}'(q_j(\bar{y}',\bar{z}') \wedge imply_{eq_i, eq_j}(\bar{x}\bar{y}\bar{z}, \bar{x}'\bar{y}'\bar{z}') )~ )~)\nonumber
\end{align}

Combining Lemmas~\ref{intermediate},
~\ref{SuppUCQ}, Propositions~\ref{comp},~\ref{val} and Corollary~\ref{imply} we get :

\begin{proposition}\label{mainUCQ}
$D \models Q_\subseteq(\bar{s}, \bar{t})$ iff $Supp(Q, D, \bar s) \subseteq Supp(Q,D, \bar t)$.
\end{proposition}

\begin{proof}
$\Rightarrow$ Assume $D \models Q_{\subseteq}(\bar{s}, \bar{t})$ and let $v \in Supp(Q, D, \bar s)$ be a valuation of $D$. By Lemma~\ref{SuppUCQ} $\exists i\bar{a} \bar{b}~D\models q_i(\bar{a}\bar{b})\wedge comp_{eq_i}(\bar{s}\bar{a}\bar{b})$ and $v(D) \models eq_i(v(\bar{s}\bar{a}\bar{b}))$. So by our assumption there exists $j,\bar{a}'\bar{b}'$ with 
$D \models~q_j(\bar{a}'\bar{b}') \wedge imply_{eq_i, eq_j}(\bar{s}\bar{a}\bar{b}, \bar{t}\bar{a}'\bar{b}')$
and by Corollary~\ref{imply} $D \models comp_{eq_j}(\bar{t}\bar{a}'\bar{b}')$.
Now let $t_1,t_2$  
such that $D \models equiv_{eq_j}(\bar{t}\bar{a}'\bar{b}',t_1,t_2)$. By $D \models~imply_{eq_i, eq_j}(\bar{s}\bar{a}\bar{b}, \bar{t}\bar{a}'\bar{b}')$ we have $D \models equiv_{eq_i}(\bar{s}\bar{a}\bar{b},t_1,t_2)$
and by Lemma~\ref{intermediate}, $v(t_1)=v(t_2)$. But then, again by Lemma~\ref{intermediate}, $v(D) \models eq_i(v(\bar{t}\bar{a}'\bar{b}'))$ and by Lemma~\ref{SuppUCQ} it follows that $v \in Supp(Q,D,\bar{t})$.

$\Leftarrow$ Assume $Supp(Q, D, \bar s) \subseteq Supp(Q,D, \bar t)$ and let $i, \bar{a}$, $\bar{b}$ with $D \models q_i(\bar{a},\bar{b})\wedge comp_{eq_i}(\bar{s},\bar{a},\bar{b})$. By Proposition~\ref{comp} there exists a valuation $v$ (that we assume w.l.o.g. to be tight) such that $v(D) \models eq_i(v(\bar{s}\bar{a}\bar{b}))$ and so by Lemma~\ref{SuppUCQ} $v \in Supp(Q,D,\bar{s})$. Hence by our assumption we also have $v \in Supp(Q,D,\bar{t})$ and so by Lemma~\ref{SuppUCQ} there exists $j,\bar{a}'\bar{b}'$ with $D \models~ q_j(\bar{a}'\bar{b}') \wedge comp_{eq_j}(\bar{t}\bar{a}'\bar{b}')$ and $v(D) \models eq_j(v(\bar{t}\bar{a}'\bar{b}'))$. As $v$ is tight, by Lemma~\ref{tight} it follows from $v(D) \models eq_j(v(\bar{t}\bar{a}'\bar{b}'))$ that $\forall v$ with $v(D) \models eq_i(v(\bar{s}\bar{a}\bar{b}))$, also $v(D) \models eq_i(v(\bar{t}\bar{a}'\bar{b}'))$. 
Now by Proposition~\ref{val} $D \models imply_{eq_i, eq_j}(\bar{s}\bar{a}\bar{b}, \bar{t}\bar{a}'\bar{b}') \vee \neg comp_{eq_i}(\bar{s},\bar{a},\bar{b})$.
But $D \models comp_{eq_i}(\bar{s},\bar{a},\bar{b})$, so $D \models \exists \bar{y}\bar{z}(~q_j(\bar{y},\bar{z}) \wedge imply_{eq_i, eq_j}(\bar{s}\bar{a}\bar{b}, \bar{t}\bar{y}\bar{z}))$. 
\end{proof}

Recall that from $Q_\subseteq$ one can easily define a first order rewriting $best_Q (\bar x)$ for best answers as in (\ref{bestq}).

\begin{theorem} \label{thm:last}
Given $Q$ a union of conjunctive queries over schema $\sigma$ and an incomplete database $D$, 
$\bar t \in \best(Q,D) \textrm{ iff } D \models best_Q(\bar t)$.
\end{theorem}

\begin{proof}
By Proposition~\ref{mainUCQ} $D \models best_Q(\bar t)$ if and only if $\forall s Supp(Q,D,\bar{t}) \subseteq Supp(Q,D,\bar{s})$ implies $Supp(Q,D,\bar{s}) \subseteq Supp(Q,D,\bar{t})$. Notice that this holds exactly whenever $\neg \exists \bar{s}$ with $Supp(Q,D,\bar{t}) \subset Supp(Q,D,\bar{s})$, i.e., whenever $\bar t \in \best(Q,D)$.
\end{proof}

\begin{example}
For $Q,D,\gamma, \gamma'$ as in Example \ref{ex} and \ref{ex::imply} :
$$Q_\subseteq(x,x') :=  \forall y_1y_2z ((R(y_1)\wedge S(y_2,z) \wedge comp_\gamma(y_1,y_2,z,x))$$ $$\rightarrow$$
$$ ~\exists y_1' y'_2 z'(R(y'_1) \wedge S(y_2',z') \wedge imply_{\gamma, \gamma'}(y_1y_2zx,y'_1y'_2z'x') )).$$

This allows to derive for instance $\supp(\qd,1)\subseteq \supp(\qd,\bot_2)$ (as observed in Example \ref{ex}). In fact
the subquery $R(y_1)\wedge S(y_2,z)\wedge comp_\gamma(y_1,y_2,z,x)$ with free variables $y_1,y_2,z,x$ selects on $D$ tuples $(1,\bot_2,\bot_2,1),(\bot_1,\bot_2,\bot_2,1)$, and no other tuple with last element $1$. Moreover as shown in Example \ref{ex::imply} 
\begin{align}
&D \models imply_{\gamma\gamma'}(\bot_1\bot_2\bot_21,1\bot_2\bot_2\bot_2)\nonumber\\
&D \models imply_{\gamma\gamma'}(1\bot_2\bot_21,1\bot_2\bot_2\bot_2)\nonumber
\end{align}
Thus
$D \models Q_\subseteq(1,\bot_2)\nonumber$.
Similarly one can show $D \models Q_\subseteq(\bot_1,\bot_2)$ and therefore $D\models best_Q(\bot_2)$.
\end{example}

As a corollary of Theorem~\ref{thm:last}, for a union of conjunctive queries $Q$ one can compute $\best(Q,D)$ by first computing the formula $best_Q(\bar x)$ from $Q$, then evaluating  $best_Q$ on $D$. Since data complexity of $\FO$ query evaluation is \dlog (and in particular {\sc AC$^0$}), this gives the following corollary :

\begin{corollary}
For each fixed union of conjunctive queries $Q$, the data complexity of \bestanswer\ 
is \dlog.
\end{corollary}

Note that it was known from \mycite{libkin:zero} that the data complexity of computing best answers for unions of conjunctive queries is polynomial time. 
In terms of combined complexity (i.e. when either $Q$, $D$ and $\bar a$ are in the input), the rewriting approach (i.e. the procedure of computing $best_Q$ from $Q$ and then evaluating $best_Q$ on $D$), can be easily shown to be in \pspace. 
In fact it is well known that a first order query $\phi$ can be evaluated on a database $D$ in space at most $qr(\phi) \log |D| + log|\phi|$, 
where $qr(\phi)$ is the quantifier rank of $\phi$. Note that although $best_Q$ has size exponential in $Q$, the quantifier rank of $best_Q$ is linear in the size of $Q$. Thus whether $\bar a \in best(Q,D)$ can be checked using space $O(|Q|, |D|)$. Moreover one can easily check that $best_Q$ can be computed from $Q$ in space polynomial in the size of $|Q|$. 
Since space bounded computations can be composed without storing the intermediate output, 
computing $best_Q$ from $Q$ and then evaluating $best_Q$ on $D$ can be done overall in PSPACE in the size of $|Q|$ and $|D|$. 
The rewriting approach thus implies a PSPACE upper bound for the combined complexity of \bestanswer\ for unions of conjunctive queries. However we can show that the problem actually stands in the third level of the polynomial hierarchy.

\begin{theorem}\label{combined-best}
For unions of conjunctive queries, combined complexity  of \bestanswer~ is $\Pi^p_3$-complete. Hardness already holds for conjunctive queries.
 \end{theorem}
 
\begin{proof} 
For membership, first note that one can check in $\Pi^p_2$ whether $Supp(Q,D,\bar a) \subseteq Supp(Q,D,\bar b)$ on input given by a database $D$, a UCQ $Q$, and tuples $\bar a$ and $\bar b$. In fact in order to check $Supp(Q,D,\bar a) \nsubseteq Supp(Q,D,\bar b)$ one guesses a valuation $v$ of $D$, then calls an NP oracle to check $v(\bar a) \in Q(v(D))$ and $v(\bar b) \notin Q(v(D))$.

On input given by a database $D$, a UCQ $Q$, and a tuple $\bar a$ one can check $\bar a \notin \best(Q,D)$ as follows.  First guess a tuple $\bar b$ over $\adom(D)$ of the same arity as $\bar a$;
then, using two calls to a $\Sigma^p_2$ oracle, check that $Supp(Q,D,\bar a) \subseteq Supp(Q,D,\bar b)$ and $Supp(Q,D,\bar b) \nsubseteq Supp(Q,D,\bar a)$.

For hardness, we reduce from $\forall\exists\forall 3DNF$, which is known to be $\Pi^p_3$-complete \mycite{sigact-news-33(3)-SU}. We take as input a $\forall\exists\forall 3DNF$-formula of the form

$$F:= \forall z_1,\ldots z_l \exists x_1 \ldots x_k \forall y_1 \ldots y_p~\bigvee^n_{i=1}conj_i$$ 
where the each $conj_i$ is a conjunction of 3 (not necessarily distinct) literals over variables $z_1,\ldots z_l, x_1, \ldots, x_k, y_1, \ldots, y_p$.

We construct a database $D_F$ with $\adom(D_F) = \{0, 1, good, bad\}\cup \{i, \bar i, \bot_i, \bar \bot_i,  | i=1..k\}$, and a conjunctive query $Q_F(z_1, ..z_l, z)$ such that $(\bar0, good)\in\best(Q_F,D_F)$ if and only if $F$ is true.

$D_F$ is of signature $\{S^4,C^2,A^2,B^3\}$ as follows :
\begin{itemize}
\item The extension of $S$ and $A$ and $B$ are fixed and do not depend on $F$ :
\begin{itemize}
\item $S$ contains tuple $(1,1,1,good)$, and tuples $(b_1,b_2,b_3,good)$ and $(b_1,b_2,b_3,bad)$ for every $b_1,b_2,b_3 \in \{0,1\}$ with $(b_1,b_2,b_3)\neq (1,1,1)$. Intuitively $S$ encodes the possible truth assignment of each disjunct of F. Note that only the satisfying assignment (i.e. (1,1,1)) appears together with the only constant $good$, all the others appear both with $good$ and $bad$.
\item $A$ contains only two tuples : $(0,1)$ and $(1,0)$. Intuitively $A$ will be used to encode truth values for pairs of literals $(w, \neg w), w \in y_1, \ldots y_p, z_1, \ldots z_l$.
\item $B$ contains tuples $(0,0, bad), (1,1, bad)$ and tuples $(b_1, b_2, good)$ and $(b_1, b_2, bad)$ for every $b_1, b_2 \in \{0,1\}$, $b_1 \neq b_2$. Intuitively $B$ encodes assignments for pairs of literals $(w, \neg w), w \in \{x_1, \ldots x_k\}$. Note that here inconsistent pairs (i.e. same truth value) are possible, but these are the only ones which do not appear together with constant $good$.
\end{itemize}
\item The extension of $C$ depends on $F$ and contains tuples $\{(\bot_i, i) | i=1..k\}$ and $\{(\bar \bot_i, \bar i) | i=1..k\}$. Intuitively a valuation $(b, i)$ (resp. $(b, \bar i)$) of one of these tuples, with $b \in \{0,1\}$, will encode truth value $b$ for the literal $x_i$ (resp, $\neg x_i$) of $F$.   
\end{itemize}

$Q_F$ is defined as follows. For each variable $w$ of $F$, the conjunctive query $Q_F$ will use variables $w$ and $\bar w$ (either quantified or free). For a literal $\alpha$ of $F$ the corresponding variable of $Q_F$ will be denoted as $enc(\alpha)$. More precisely if $\alpha = w$ is a positive literal then $enc(\alpha):=w$, otherwise if $\alpha = \neg w$  then $enc(\alpha) := \bar w$.
$$
\begin{array}{ll}
 Q_F(z_1, \ldots z_l, z) := & \exists x_1, \ldots x_k, \bar x_1, \ldots \bar x_k, y_1, \ldots y_p, \bar y_1, \ldots \bar y_p, \bar z_1, \ldots \bar z_p  \\ 
 \\
&  \bigwedge_{i=1,..k} B(x_i, \bar x_i, z) \wedge \bigwedge_{i=1,..p} A(y_i, \bar y_i) ~ ~\wedge \bigwedge_{i=1,..l} A(z_i, \bar z_i) ~ ~\wedge \\
\\
& \bigwedge_{i=1,..k} (C(x_i, i) \wedge C(\bar x_i, \bar i)) ~ ~\wedge \\
\\
& \bigwedge_{(\alpha_1 \wedge \alpha_2 \wedge \alpha_3)\in F} S(enc(\alpha_1), enc(\alpha_2), enc(\alpha_3), z)
 \end{array}
 $$
We can prove that all tuples of the form $(\bar t, good)$ (which we refer to as good tuples) have the same support. This is given by the set of all consistent boolean valuations (i.e. valuations of $\bot_i, \bar\bot_i$ in  $\{0,1\}$ such that $v(\bot_i) \neq v(\bar \bot_i)$ for all $i$).
Moreover we can prove that if there exists a $(\bar t, bad)$ whose support contains all consistent boolean valuations then the support of $(\bar t, bad)$ strictly contains the support of good tuples.
Therefore any good tuple (including $(\bar 0, good)$) is a best answer iff for all tuples $\bar t$ there exists a consistent boolean valuation which is not in the support of $(\bar t, bad)$. We can finally show that the last holds iff $F$ is true.
\end{proof}

Therefore under standard complexity theoretic assumptions, 
our rewriting approach is not optimal in terms of combined complexity, as it is often the case with generic approaches. However it has the advantage of exploiting standard \FO\ query evaluation, which despite the PSPACE combined complexity, is highly optimised in database systems and works well in practice.

%
%
\label{ijcai_rewriting}

\section{Future work}

Our rewriting techniques are closer to a practical implementation than the previous tableau based method from \mycite{gheerbrant-libkin:xml}. This is due to their expressibility in recursive SQL (or even non-recursive in the case of Theorems \ref{certainFO} and \ref{thm:last}). However, while theoretically feasible, an actual implementation will need additional techniques to achieve acceptable performance. To see why, notice that the first rule in the definition of $equiv_\gamma$ creates a cross product over the full active domain, i.e., the set of all elements that appeared in the database. This of course will be prohibitively large. While this may appear to be a significant obstacle, a similar situation with computing or approximating certain answers is not new in the literature. For instance, the first approximation scheme for certain answers to SQL queries that appeared in \mycite{tods16} has done exactly the same, and generated very large Cartesian products even for simple queries with negation. Nonetheless, an alternative was found quickly \mycite{pods16} that completely avoided the need for such expensive queries, and it was shown to work well on several TPC-H queries. Thus, looking for a practical and implementable rewriting is one of the possible directions for future work.

As another open problem, we note that the query for which we have shown certain answers to be non-rewritable in FO has DLOGSPACE data complexity. Indeed the problem is essentially reachability over trees, which can be easily encoded using deterministic transitive closure \mycite{Immerman87}. To express DLOGSPACE problems, we need a language weaker than Datalog with negation. Thus, it is natural to ask whether a low complexity Datalog fragment would be sufficient to express rewritings of BCCQ, or a separating example that is PTIME-complete can be found.

Another direction would be to investigate how our techniques can be extended to different semantics of incompleteness.
We used here the closed-world semantics \mycite{abiteboul-et-al:dbbook,imielinski-et-al:incomp,vandermeyden:incomp-survey}, in which data values are the only missing information,
but there are other possible semantics,  e.g. needed in order to cope with data inconsistencies 
\mycite{cali-et-al:constraints}, 
where query rewritings could still be found. Quantitative variations of the notion of certainty as proposed in \mycite{libkin:zero} could also be investigated.

\bibliographystyle{tlplike}
\bibliography{biblio}
\end{document}